\documentclass[journal]{IEEEtran}

\IEEEoverridecommandlockouts    
\usepackage{psfrag,bbm}
\usepackage{amsmath,amsthm}
\usepackage{amssymb}

\usepackage{amsfonts}
\usepackage{latexsym}
\usepackage{graphicx}

\usepackage{colortbl}
\usepackage{fancyhdr}
\usepackage{epstopdf}
\usepackage{float}
\usepackage{subfigure}

\usepackage{enumerate}
\usepackage{multicol}

\usepackage[usenames,dvipsnames,svgnames,table]{xcolor}
\newcommand{\Cov} {\mathbb{V}\,}

\newenvironment{ps}
  {\left[\begin{smallmatrix}}
  {\end{smallmatrix}\right]}

\def \gr{\gamma_r}

\newcommand{\Var} {\mathbb{V}\,}
\newcommand{\R} {\mathbb{R}}
\newcommand{\Z} {\mathbb{N}}
\newcommand{\N} {\mathbb{N}}
\newcommand{\sigmau} {u}

  \newtheorem{assumption}{Assumption}

      \newtheorem{proposition}{Proposition}

      \newtheorem{corollary}{Corollary}
  
\newtheorem{examplex}{Example}
  \newenvironment{example}
  {\pushQED{\qed}\examplex}
  {\popQED\endexamplex}
  
\newtheorem{remarkx}{Remark}
  \newenvironment{remark}
  {\pushQED{\qed}\remarkx}
  {\popQED\endremarkx}

    \newtheorem{definition}{Definition}
  
  \newtheorem{theoremx}{Theorem}
  \newenvironment{theorem}
  {\pushQED{\qed}\theoremx}
  {\popQED\endtheoremx}

\newcommand {\be}{\begin{equation}}
\newcommand {\ee}{\end{equation}}

\newcommand {\bes}{\begin{equation*}}
\newcommand {\ees}{\end{equation*}}
\newcommand {\bev}{\begin{equation}}
\newcommand {\eev}{\end{equation}}

\newcommand{\intercept}{v}

\newcommand{\Rc}{Q}

\usepackage{multirow}
\usepackage{rotating}

\usepackage{soul}

\IEEEoverridecommandlockouts

\usepackage{soul}
\usepackage[ruled]{algorithm2e}

\usepackage{cuted}

\begin{document}

\title{Computing the projected reachable set of switched affine systems: an application to systems biology  }

\author{Francesca Parise, Maria Elena Valcher and John Lygeros
\thanks{F. Parise is with the Laboratory for Information and Decision Systems, MIT, Cambridge, MA: {\tt\small parisef@mit.edu},
 M.E. Valcher is with  the Department of Information Engineering, University of Padova, Italy:
        {\tt\small meme@dei.unipd.it} and J. Lygeros is with the Automatic Control Laboratory, ETH, Zurich,  Switzerland:  {\tt\small lygeros@control.ee.ethz.ch}. We thank M. Khammash for allowing us to perform the experiments in the CTSB Laboratory, ETH, and J. Ruess and A.M. Argeitis for their help in
collecting the data of Fig. 5a). This work was supported by the SNSF grant number P2EZP2 168812. }
         
}
\maketitle

\begin{abstract}
A fundamental question in systems biology   is what combinations of mean and variance of the species present in a stochastic biochemical reaction network are attainable by perturbing the system~with an external signal. To address this question,  
we show that the moments evolution in any generic network can be either approximated or, under suitable assumptions, computed exactly as the solution of a switched affine system. Motivated by this application, we  propose a new method to  approximate the reachable set of  switched affine  systems. 
A remarkable feature of our approach is that it allows one to easily compute projections of the reachable set for pairs of moments of interest, without requiring the computation of the full reachable set, which can be prohibitive for large networks. As a second contribution, we also show how to select the external signal in order to maximize the probability of reaching a target set.
To illustrate  the  method we  study a renown  model of controlled gene expression and we derive estimates  of the reachable set, for the protein mean and variance, that are more accurate than those available in the literature and  consistent with experimental data.
\end{abstract}

\section{Introduction}

One of the most impressive results achieved by synthetic biology in the last decade is the introduction of externally controllable modules in biochemical reaction networks. These are biochemical circuits that react to external signals, as for example light pulses \cite{milias2011silico, parise2015guides, Olson2014a} or concentration signals \cite{Batt2012, menolascina2011analysis}, allowing researchers to influence and possibly control the behavior of  cells \textit{in vivo}.  To fully exploit these tools, it is  important  to first understand what range of behaviors they can exhibit under different choices of the external signal.
For deterministic systems, this  amounts to computing the set of states that can be reached by the controlled system~trajectories starting from a known initial configuration   \cite{batt2008symbolic, chabrier2003symbolic}. Since chemical species are often present in low copy numbers inside the cell, biochemical reaction networks can
 however be
 inherently stochastic \cite{gillespie1992rigorous}. 
 In other words, if we apply the same signal to a population of identical cells, then every cell will have a different evolution (with different likelihood), requiring  a probabilistic analysis.
 
  If we interpret each cell has an independent realization, we can then study the effect of the external signal on  a population of cells  by characterizing how such a signal influences the moments of the underlying stochastic process. Specifically, 
in this paper we  pose the following  question: 
\begin{center}
\textit{``What combinations of  moments of the stochastic process can be achieved by applying the external signal?''}
\end{center}
 
 This approach is motivated for example by biotechnology applications, where one would like to control the average behavior of the cells in large populations, instead of each cell individually. More on the theoretical side, this  perspective can be useful to investigate fundamental questions on noise suppression in biochemical reaction networks, as  in \cite{lestas2010fundamental}. 

The cornerstone of our approach is the observation that while the number of copies in each cell is stochastic, the evolution of the moments is deterministic and can either  be described or approximated by a switched affine 
system.  
Consequently, the above question can be reformulated as a reachability problem in the moment space.
Computing the exact reachable set of a switched affine  
system~is in general far from trivial, see \cite{sun2006switched,altafini2002reachable}.  We thus start our analysis  by proposing a new method to approximate the reachable  set of a switched affine system. This  is an extension of the hyperplane method for linear systems suggested in \cite{krogh} and is of interest on its own. We then show how to apply the proposed approach  to  biochemical reaction networks by distinguishing two cases:
\begin{enumerate}
\item If all the reactions follow the laws of mass action kinetics and are at most of order one, the system~of moments equations is switched affine. Consequently, for this class of networks, the above question can be solved by directly  applying the newly suggested hyperplane method in the moments space;
\item For all other reaction networks the moments equations are in general non-closed (i.e., the evolution of mean and variance depends on higher order moments). We show however that the evolution of the probability of being in a given state can be described by an \textit{infinite} dimensional switched system~and that the desired moments can be computed as the output of such  system. We then show: i)  How to approximate such  an infinite dimensional system~with a finite dimensional one, by extending the finite state projection  method \cite{fsp} to controllable networks, ii) How to compute the reachable set of the finite dimensional system~by applying the newly suggested hyperplane method in the probability space, and iii) How to recover an approximation of the original reachable set from the reachable set of the finite dimensional system.   
\end{enumerate}

In the last part of the paper, we change perspective and, instead of focusing on population properties,  we consider  the behaviour of a single cell (i.e., a single realization of the process), given a fixed initial condition or an initial probability distribution. Such perspective has been commonly employed for the case without external signals, see e.g. \cite{baier2003model, abate2010approximate, el2006advanced, fsp}. Our objective is to  show how the external signal can be used to control single cell realizations by posing the following question
\begin{center}
\textit{``What external signal should be applied to maximize the probability that the  cell trajectory reaches a prespecified subset of the state space at the end of the experiment?''}
\end{center}
We  show that such a problem can be  addressed by using similar tools as those derived for the population analysis.

\subsubsection*{Comparison with the literature}
 A vast literature has been devoted to the analysis of the reachable set of piecewise-affine systems in the context of hybrid systems, see e.g. \cite{sun2006switched,alur1993hybrid, alur2000discrete, koutsoukos2003safety, ghosh2004symbolic, habets2006reachability, hamadeh2008reachability} among many. 
Our results are different because we exploit the specific structure of the problem at hand, that is,  the fact that the switching signal is a control variable and that the  dynamics in each mode  are autonomous and affine.  
In other words, we consider switched affine systems for which the switching signal is the only control action. 
We  also note  that many different methods have been proposed in the literature to compute the reachable set of  generic nonlinear systems. Among these there are level set methods \cite{mitchell2005toolbox},   ellipsoidal methods \cite{kurzhanskiui1997ellipsoidal} and sensitivity based methods \cite{donze2007systematic}. For example, we became aware at the time of submission that the authors of \cite{lakatos2016control} extended our previous works \cite{parise2014reachable,parise2015reachable} by suggesting the use of ellipsoidal methods. It is  important to stress that
the choice of a method that scales well with the system~ size  is essential in our context, since
biochemical networks are typically very large.  Moreover,  biologists are often interested in analyzing the behavior of only  a few chemical species of the possibly many involved in the network. Consequently, one is  usually interested in computing the projection of the reachable set (which is  a high-dimensional object) on some low-dimensional space of interest. The hyperplane method that we propose stands out in this  respect since, by using a method tailored for switched systems,  it allows one to compute directly the  projections of the reachable set, without requiring the computation of the full high-dimensional reachable set first. We thus avoide the curse of dimensionality that characterises all the previously mentioned methods.
We  note that part of the results of this paper appeared in our previous works \cite{parise2015reachable,parise2016reachability}. Specifically, 
in \cite{parise2015reachable} we first suggest the use of the hyperplane method  to compute the reachable set of biochemical networks with \textit{linear} moment equations, which we then adapted in  \cite{parise2015reachable}  to the case of \textit{switched affine} moment equations. As better detailed in Section \ref{sec:th:mom}, the assumptions made both in \cite{parise2015reachable} and \cite{parise2016reachability} do not allow for bimolecular reactions, which are instead present in the vast majority of biochemical networks. The key contribution of this paper is the generalisation of our  analysis to \textit{any} biochemical network by using the approach described in point 2) above.  The analysis of single cell realizations is also entirely new.

\subsubsection*{Outline}
In Section~\ref{sec:r_tol} we present  the hyperplane method. In Section~\ref{sec:r_tol_af} we review how to compute the hyperplane constants for linear systems, while  in Section~\ref{sec:r_tol_sw} we propose a new procedure for switched affine systems. In Section~\ref{CME} we introduce stochastic biochemical reaction networks and the controlled chemical master equation (CME). Additionally, we recap how to derive the moments equations from the CME (Section \ref{sec:th:mom}) and we derive an extension of the finite state projection  method to controlled biochemical networks (Section \ref{sec:th:fsp}). In Section~\ref{sec:reach} we show how to compute the reachable set of biochemical networks 
and in Section~\ref{sec:individual} we derive the results on single cell realizations. 
Section~\ref{gene} illustrates our theoretical results on a gene expression case study.

\subsubsection*{Notation}
Given $a<b\in\mathbb{N}$, we  set $\mathbb{N}[a,b]:=\{a,a+1,\ldots,b\}$.
Given a set $\mathcal{S}$, the symbol $\partial \mathcal{S}$ denotes its boundary, $\textup{conv}(\mathcal{S})$ its convex hull and $|\mathcal{S}|$ its cardinality.  For a vector $x\in\mathbb{R}^n$, $x_p:=\left[x\right]_p$ denotes its $p$th component,  $|x|:=[|x_1|^\top,\ldots,|x_n|^\top]^\top$ and $\|x\|_\infty:=\max_{p=1,2,\dots, n}|x_p|$ denotes the infinity norm.   $\mathbbm{1}$ denotes a vector of all ones. Given two random variables $Z_1,Z_2$, we denote  by $\mathbb{V}[Z_1]$  and $\mathbb{V}[Z_1,Z_2]$ their variance and covariance, respectively.

\section{Reachability tools}
\label{sec:r_tol}

\subsection{The reachable set and the hyperplane method}
\label{sec:r_tol_def}
Consider  the $n$-dimensional nonlinear control system
\begin{align}
\dot x(t)=f(x(t) ,\sigma(t)), \quad t\ge 0,
\label{eq:nnl}
\end{align}
 where $x$ is the $n$-dimensional state and $\sigma$ the $m$-dimensional  input function.  Set a final time $T>0$ and let ${\mathcal S}$ be  {\em the set of admissible input functions } that we assume to be a subset of the set of all measurable functions that map $[0,T]$  into $\R^m$.  We  assume that  the function $f:\mathbb{R}^n\times \mathbb{R}^m\rightarrow \mathbb{R}^n$ is such that, for every initial condition $x(0)\in \mathbb{R}^n$ and   every  input function  $\sigma\in\mathcal{S}$, the solution of  \eqref{eq:nnl}, denoted by  $x(t;  x(0), \sigma),  t\ge0,$ is well defined and unique at every time $t\ge 0$. 
The reachable set of system~\eqref{eq:nnl} at time $T$ is defined as the set of all states $x\in\mathbb{R}^n$ that   can be reached at time $T$, starting from $x(0)$, by using an admissible  input function $\sigma\in {\mathcal S}$.

\begin{definition}[{Reachable set at time $T$}] The  reachable set   
at time $T>0$  from $x(0)=x_0$,   for  system~\eqref{eq:nnl} with  admissible input set ${\mathcal S}$, is
\begin{equation}  \mathcal{R}_T(x_0):=\{x\in\mathbb{R}^n \mid \exists\ \sigma\in  {\mathcal S} : x=x(T; x_0, \sigma)
\}.
 \end{equation}
\end{definition}
From now on we   will assume that the set $\mathcal{R}_T(x_0)$ is compact, since this will be the case for all the systems of interest analysed in the following.
Computing such   a reachable set for  nonlinear systems is in general a very difficult task.
For the case of linear systems with bounded inputs  a method to construct an outer approximation of $\mathcal{R}_T(x_0)$ as  the intersection of a family of half-spaces that are tangent to its boundary (see Fig.~\ref{fig:hm}) was  proposed in~\cite{krogh}. 

\begin{figure}[h]
\begin{center}
\includegraphics[width=0.35\textwidth]{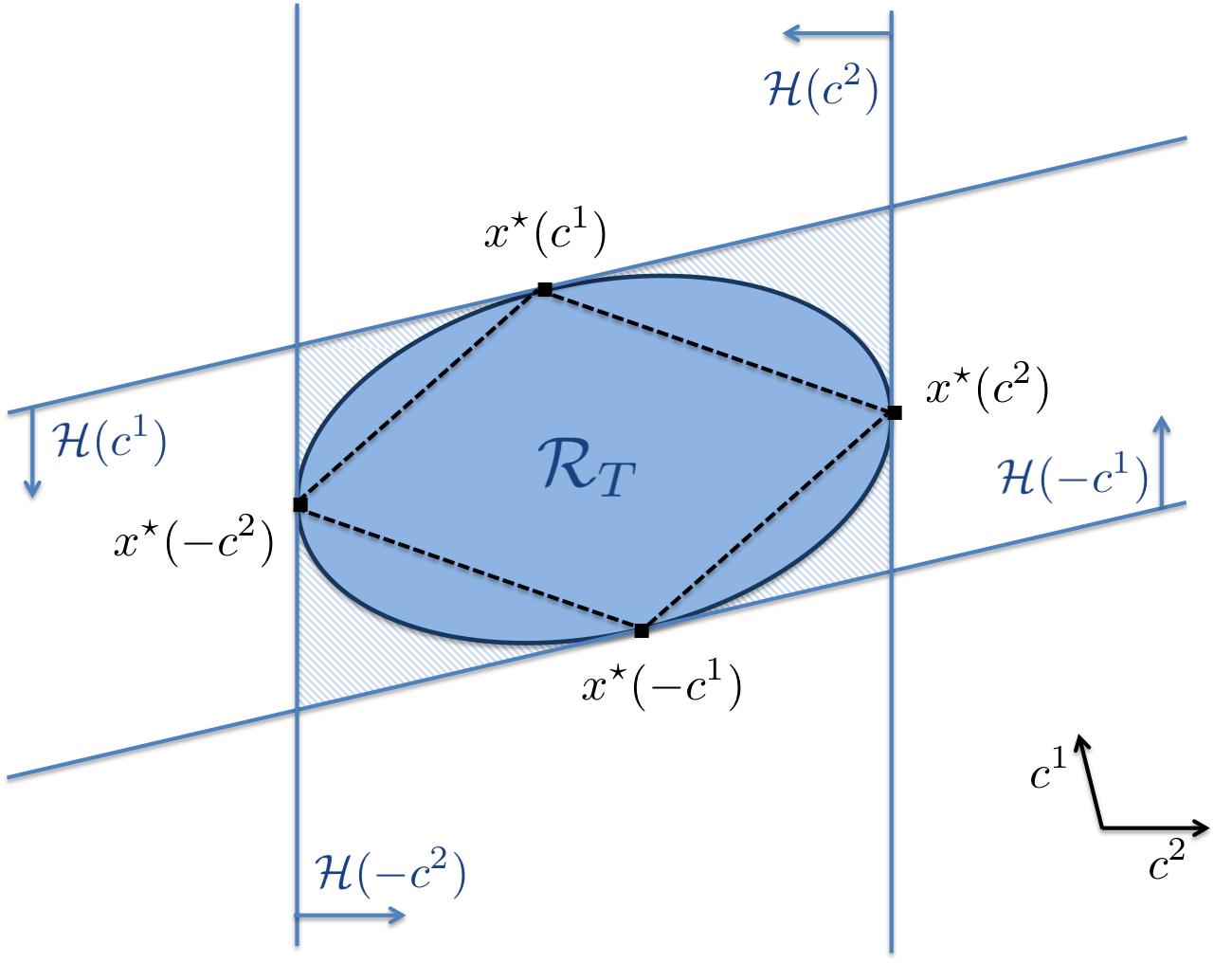}
\end{center}
\vspace{-0.5cm}
\caption{Illustration of the hyperplane method for a  \textit{convex} reachable set  $\mathcal{R}_T(x_0)$ (in blue). The external parallelogram is the outer approximation, the region in between the dotted lines is the inner approximation.  }
\label{fig:hm}
\end{figure}

We present here a generalisation  of this method to system~\eqref{eq:nnl}.
For a given direction $c\in\mathbb{R}^n$, let us define
\begin{align}\label{eq:intercept_general}
\intercept_T(c):=\max_{x\in \mathcal{R}_{T}(x_0)} c^\top x,
\end{align}
where, for simplicity, we omitted the dependence of $\intercept_T(c)$ on the initial condition $x_0$.
Let 
\begin{align}\label{eq:half_space}
\textup{H}_T(c)&:=\{ x\in\mathbb{R}^n \mid {c}^\top x= \intercept_T(c) \}
\end{align}
 be the corresponding  hyperplane. By definition of the constant $\intercept_T(c)$,  the associated half-space 
\begin{align}\label{eq:half_space}
\mathcal{H}_T(c)&:=\{ x\in\mathbb{R}^n \mid {c}^\top x\le \intercept_T(c) \}
\end{align}
is a superset of  $\mathcal{R}_T(x_0).$ We note that if $\partial \mathcal{R}_T(x_0)$ is smooth, then $\textup{H}_T(c)$ is  the tangent plane to $\partial \mathcal{R}_T(x_0)$. 
By   evaluating the above hyperplanes and half-spaces
for various  directions,
one can construct an outer approximation of the   reachable set, as illustrated in the next theorem. If the reachable set is convex then an inner approximation can also be derived.

\begin{theorem}[The hyperplane method \cite{krogh}] 
\label{thm:hm1}
Given system~\eqref{eq:nnl},  an initial condition $x_0\in {\mathbb R}^n$, a fixed time $T>0$,  an integer number  $D\ge2$, and a set of  $D$ directions $\mathcal{C}:=\{c^1,\ldots,c^D\}$, define the half-spaces 
$
\mathcal{H}_T(c^d)$
as in \eqref{eq:half_space}, for   $d=1,\ldots,D$. 
 \begin{enumerate}
\item The  set
\begin{align*}
\mathcal{R}_T^{out}(x_0):=\cap_{d=1}^D  \mathcal{H}_T(c^d) \quad
\end{align*}
is an outer  approximation of the reachable set $\mathcal{R}_T(x_0)$ at time $T$  starting from $x_0$.

\item If the set $\mathcal{R}_T(x_0)$ is convex and for each $d= 1, 2, \dots, D,$ we select a (tangent) point 
\begin{align} \label{x}
x_T^{\star}(c^d)&\in \mathcal{R}_T(x_0) \cap \textup{H}_T(c^d) 
\end{align} 
then the set 
\begin{align*}
\mathcal{R}_T^{in}(x_0):=\textup{conv} \left(\{x_T^{\star}(c^d), d= 1,2, \dots, D\}\right)
\end{align*}
is an inner approximation of the reachable set $\mathcal{R}_T(x_0)$ at time $T$  starting from $x_0$.  \qedhere
\end{enumerate}
\end{theorem}

\begin{remark} 
We note that by construction the outer approximation $\mathcal{R}_T^{out}(x_0)$ is a convex object.  Specifically, when the number of hyperplanes tends to infinity $\mathcal{R}_T^{out}(x_0)$ coincides with the convex hull of $\mathcal{R}_T(x_0)$. Similarly, for \textup{any}  set $\mathcal{R}_T(x_0)$, the set $\mathcal{R}_T^{in}(x_0)$ is an inner approximation of the convex hull of $\mathcal{R}_T(x_0)$. However, the inner approximation of the convex hull of a set is an inner approximation of the set itself only if such set is convex, as assumed in the previous theorem.
\end{remark}

The main advantage of this method is that hyperplanes are very easy objects to handle and visualise. The main disadvantage is that 
the higher the dimension $n$ of the state space,   the higher    in general is  the number of  directions $D$ required to obtain a  good characterisation of the reachable set. In the next subsection we  show how to avoid this curse of dimensionality,   in cases when  only  the projection of the reachable set on a plane of interest is needed. 

\subsection{The output reachable set}

  Let the output of system~\eqref{eq:nnl} be
\begin{equation}
y(t)=Lx(t),
\label{outputA}
\end{equation}
for  $L\in\R^{p\times n}$,
and the output reachable set be  the set of all output values   that can be generated at time $T$ from $x(0) =x_0$, by using an admissible input function
$\sigma\in {\mathcal S}$.
\smallskip

\begin{definition}[{Output reachable set at time $T$}] The output reachable set $\mathcal{R}_{T}^y(x_0)$ {\em from $x_0$ at time} $T>0$,   for system~ \eqref{eq:nnl} with  admissible input set $ {\mathcal S}$ and output    as in \eqref{outputA}, is 
\begin{align*}  
\hspace{0.7cm} 
\mathcal{R}_{T}^y(x_0):=\{y\in\mathbb{R}^p \mid \exists\ x \in \mathcal{R}_T(x_0) :  y=Lx
\}. \qedhere
 \end{align*}
\end{definition}

For simplicity, in the following we restrict our discussion to the case of a two-dimentional output vector, that is 
\begin{equation}
y(t)=Lx(t)=\begin{bmatrix}l_1^\top x(t)\cr  l_2^\top x(t)\end{bmatrix} \in\mathbb{R}^2,
\label{output}
\end{equation}
for some $l_1,l_2\in\R^n$, the generalization to higher dimentions is however immediate. Note that,  for any pair of  indices $i,j \in \{1,\ldots,n\}, i\neq j$,  one can  recover the projection of the reachable set $\mathcal{R}_T(x_0)$   onto   an $(x_i,x_j)$-plane of interest by imposing  $l_1=e_i$ and $l_2=e_j$. The two-dimentional output vector case can therefore be  applied   to study the relation between the mean behavior of two species or between mean and variance of a single species in large biochemical networks.

In the following theorem  we show that inner and outer approximations of  $\mathcal{R}_{T}^y(x_0)$ can be efficiently computed by selecting only hyperplanes that are perpendicular to the plane of interest. 

\begin{theorem}[Projection on a two dimensional subspace] 
\label{thm:hm2}
Consider system~\eqref{eq:nnl}, with output \eqref{output}    and initial condition $x_0\in {\mathbb R}^n$. Let  $T>0$ be a fixed time, $D\ge 2$  an integer number and choose $D$ values $\gamma^d\in\mathbb{R}$. Set  $c^d:=l_2-\gamma^d l_1 \in\mathbb{R}^n$ and 
\begin{align*}
\mathcal{H}^{y}_{T}(\gamma^d)&:=\{ y \in\mathbb{R}^2 \mid y_2 \le \gamma^d y_1+ \intercept_T(c^d) \},
\end{align*}
where $\intercept_T(c^d)$ is   as in \eqref{eq:intercept_general}. Set
$
y_{T}^{\star}(\gamma^d)\textstyle :=L x_{T}^{\star}(c^d),
$
where $x_{T}^{\star}(c^d)$ is defined  as in \eqref{x}.
Then the set 
\begin{align}
\mathcal{R}_{T}^{y,out}(x_0):=\cap_{d=1}^D  \mathcal{H}_{T}^{y}(\gamma^d) \label{out}
\end{align}
is an outer approximation of  $\mathcal{R}_{T}^y(x_0)$. Moreover, if  $\mathcal{R}_{T}(x_0)$ is convex then the set
\begin{align}
\mathcal{R}_{T}^{y,in}(x_0):=\textup{conv} \left(\{y_T^{\star}(\gamma^d), d= 1,2, \dots, D\}\right)
\end{align}
is an inner approximation of $\mathcal{R}_{T}^y(x_0)$. 
 \end{theorem}
 
 \begin{proof}
By definition, for   any $\bar y \in \mathcal{R}_{T}^y(x_0)  $ there exists an $\bar x\in \mathcal{R}_T(x_0)$ such that $\bar y^\top=[ l_1^\top \bar x,l_2^\top \bar x]$.
By Theorem \ref{thm:hm1}, for any direction $c^d$ it holds that $\mathcal{R}_T(x_0)\subset \mathcal{H}_T( c^d)$. Consequently, $\bar x\in \mathcal{R}_T(x_0)$ implies $\bar x\in \mathcal{H}_T( c^d)$.
By substituting the definition of $c^d$ given in the statement we get  
\begin{align*}
&\bar x\in \mathcal{H}_T( c^d)
\Leftrightarrow (c^d)^\top \bar x \le \intercept_T(c^d) \Leftrightarrow \\
&\Leftrightarrow (l_2-\gamma^d l_1)^\top \bar x \le \intercept_T(c^d) \Leftrightarrow  l_2^\top \bar x \le \gamma^d l_1^\top \bar x +\intercept_T(c^d).
\end{align*}
The last inequality implies $\bar y^\top=[ l_1^\top \bar x,l_2^\top \bar x] \in \mathcal{H}^y_T(\gamma^d)$.
Consequently, $\mathcal{R}_{T}^y(x_0)\subseteq \mathcal{H}^y_T(\gamma^d)$ for any $\gamma^d$ and therefore $\mathcal{R}_{T}^y(x_0)\subseteq  \mathcal{R}^{y,out}_T(x_0)$. If $\mathcal{R}_T(x_0)$ is convex, then  $\mathcal{R}_{T}^y(x_0)$ is convex as well. The  points $y^\star_T(\gamma^d)$ belong to $\mathcal{R}_{T}^y(x_0)$ by construction. Consequently, by convexity, it must hold that 
$\mathcal{R}^{y,in}_T(x_0) \subseteq  \mathcal{R}_{T}^y(x_0)$.
 \end{proof}

\section{Computing the tangent hyperplanes}
\label{computing}
The success of the hyperplane method hinges on the possibility of efficiently evaluating, for any given direction $c$, the  constant $\intercept_T(c)$ in  \eqref{eq:intercept_general}. 
 Note that this problem  is equivalent to the following finite time optimal control problem
\begin{align}
\intercept_T(c):=\max_{\sigma\in  {\mathcal S}} &\quad c^\top x(T)\label{intercept}\\
 \mbox{s.t.}&\quad  \dot x(t)=f(x(t) ,\sigma(t)),\quad \forall t\in[0,T],\notag\\
 &\quad x(0)=x_0 \notag.
\end{align}
In the rest of this section, we aim at solving \eqref{intercept}. To this end, we start by  recalling 
 the linear case, for which the hyperplane  method was originally derived in \cite{krogh}.

\subsection{Linear systems with bounded input}
\label{sec:r_tol_af}

The hyperplane method was originally proposed  for linear   systems with bounded inputs
\begin{align}
\dot x(t)=Ax(t)+B\sigma(t),
\label{eq:lin}
\end{align}
where $x(t)\in\mathbb{R}^n$, $A\in\mathbb{R}^{n\times n}$,  $B\in\mathbb{R}^{n\times m}$ and $\sigma(t)\in\R^m$. 
Since biological signals are non-negative and bounded, we here make the following assumption on the input set $\mathcal{S}$.
 \begin{assumption} 
The input function  $\sigma$ belongs to the admissible set  
$\mathcal{S}_\Sigma:=\{\sigma \mid \sigma(t)\in\Sigma, \forall t\in[0,T]\},$
 where  $\Sigma = \Sigma_1\times \ldots\times\Sigma_m$. 
  Moreover,  there exist  $\bar \sigma_r>0, r \in\N[ 1, m],$ such that  
either (a) every set  $\Sigma_r$ is the  interval $\Sigma^c_r:=[0,\bar \sigma_r]$ (continuous and bounded input set),  
or (b) for every  set  $\Sigma_r$  there exists $2\le q_r<+\infty$ such that $\Sigma^d_r:=\left\{ 0=\sigma_r^1<\sigma_r^2<\ldots<\sigma_r^{q_r}=\bar{\sigma}_r \right\}\subset \mathbb{R}_{\ge0}$ (finite  input set). We  set $\Sigma^c:= \Sigma^c_1\times \ldots\times\Sigma^c_m$, $\Sigma^d:= \Sigma^d_1\times \ldots\times\Sigma^d_m$,  and denote by $\mathcal{S}_{\Sigma^c}$ and $\mathcal{S}_{\Sigma^d}$ the  corresponding admissible sets.
 \label{ass:linear_input}
 \end{assumption}

\noindent In the case of a continuous and bounded input set, i.e. under Assumption \ref{ass:linear_input}-(a), it was shown in \cite{krogh} that  it is possible to solve the control problem in \eqref{intercept} in closed form by  using the Maximum Principle \cite{liberzon2011calculus}.
\begin{proposition}[Tangent hyperplanes for linear systems   with bounded and continuous inputs]
\label{tangent}
Consider system~\eqref{eq:lin} and suppose that Assumption \ref{ass:linear_input}-(a) holds.  Define the following admissible input function, expressed component-wise for every 
$r$th entry, $r=1,\hdots,m$, as
\begin{align}
\sigma_r^{\star}(t):&= \begin{cases}
   \bar{\sigma}_r & \text{if } \quad c^\top e^{A(T-t)}b_r> 0; \\
  0 & \text{if }\quad c^\top e^{A(T-t)}b_r< 0; \\
   0\le \sigma_r \le \bar{\sigma}_r & \text{if }\quad c^\top e^{A(T-t)}b_r= 0;
  \end{cases}
\end{align}
where $b_r$ denotes the $r$th column of $B$.
Then
\begin{align}
\textstyle \intercept_T(c)&=\textstyle c^\top e^{AT}x_0+\sum_{r=1}^m \bar{\sigma}_r  \int_{0}^T \left[ c^\top e^{A(T-t)}b_r \right]_+ dt,
\label{var1}
\end{align}
 where $[g(t)]_+$ denotes the positive part of the function, namely $[g(t)]_+=g(t)$ when $g(t)>0$ and zero otherwise. Suppose additionally that the pair $(A,b_r)$ is reachable, for every  $r\in\N[1,m].$
Then there exists  no interval $[\tau_1,\tau_2]$, with $0\le \tau_1<\tau_2\le T$, such that $c^\top e^{A(T-t)}b_r= 0$ for every $t\in [\tau_1,\tau_2]$. Consequently,   a tangent point    can be obtained as 
\begin{align}\label{var2}
\textstyle x^{\star}_T(c)&:=\textstyle e^{AT}x_0+\int_{0}^T e^{A(T-t)}B\sigma^{\star}(t)dt. \qedhere 
\end{align} 
\hfill{$ \square$}
\end{proposition}

The proof follows the same lines as  \cite[Lemma 2.1 and Theorem 2.1]{krogh} and is omitted for the sake of  brevity.

 By using the explicit characterisation given in Proposition~\ref{tangent} together with Theorems \ref{thm:hm1} and \ref{thm:hm2}, one can efficiently  construct both an inner and an outer approximation of  the (output) reachable set for  linear systems with  \textit{continuous and bounded  input set} $\Sigma^c$, as summarised in the next corollary. Therein we also show  how  the same result   can be  extended  to \textit{finite input sets} $\Sigma^d$.
 \begin{corollary}[The hyperplane method for linear systems] \label{cor:linear}
 Consider system~\eqref{eq:lin} and suppose that either Assumption \ref{ass:linear_input}-(a) or  Assumption \ref{ass:linear_input}-(b) holds. Let $\intercept_T(c^d)$ and $x^{\star}_T(c^d) $ be computed as in  \eqref{var1} and \eqref{var2}.
 Then $\mathcal{R}^{out}_T(x_0)$ and $\mathcal{R}^{in}_T(x_0)$ ($ \mathcal{R}^{y,out}_{T}(x_0)$ and $\mathcal{R}^{y,in}_{T}(x_0)$, resp.) as defined in Theorem \ref{thm:hm1} (Theorem~\ref{thm:hm2}, resp.)  are outer and inner approximations of $\mathcal{R}_T(x_0)$ (of $\mathcal{R}^{y}_{T}(x_0)$, resp.).
  \end{corollary}
  
  \begin{proof}
  In the case of continuous and bounded input, that is, under Assumption \ref{ass:linear_input}-(a), the reachable set $\mathcal{R}_T(x_0)$ is convex and the statement is  a trivial consequence on Theorems \ref{thm:hm1} and \ref{thm:hm2} and Proposition \ref{tangent}. We here show that the same result holds also under Assumption \ref{ass:linear_input}-(b). The proof of this second part follows from the fact that the  reachable set $\mathcal{R}_{T}^c(x_0)$, obtained by using the continuous input set $\Sigma^c$, and the  reachable set $\mathcal{R}^d_T(x_0)$, obtained by using the discrete input set $\Sigma^{d}$, coincide. To prove this, let   $\mathcal{R}^{bb}_T(x_0)$ be the  reachable set obtained using  $\Sigma^{bb}_r:=\{0,\bar\sigma_r\}$ for any $r$, that is, the set of vertices of $\Sigma^c$.  
Consider now an arbitrary point $\bar x\in \mathcal{R}_T^c(x_0)$, which is a compact set. By definition there exists an admissible input function  in $\Sigma^c$ that steers $x_0$ to $\bar x$ in time $T$. Since $\Sigma^c$ is a convex polyhedron, by \cite[Theorem 8.1.2]{sussmann}, system~\eqref{eq:lin} with input set $\Sigma^c$ has the bang-bang with  bound on the number of switchings (BBNS) property. That is, for each $\bar x\in \mathcal{R}_T^c(x_0)$  there exists a bang-bang input function in $\Sigma^{bb}$ that reaches $\bar x$  in the same time $T$ with a finite number of discontinuities. Thus $\bar x\in \mathcal{R}_T^{bb}(x_0)$.  Since this is true for any $\bar x\in \mathcal{R}_T^c(x_0)$, we get   $\mathcal{R}_{T}^c(x_0)\subseteq \mathcal{R}^{bb}_T(x_0)$. From $\Sigma^{bb}\subseteq\Sigma^d\subseteq\Sigma^c$ we get $ \mathcal{R}^{bb}_T(x_0)\subseteq \mathcal{R}^{d}_T(x_0)\subseteq \mathcal{R}_T^c(x_0)$,  concluding the proof.  \end{proof}

\subsection{Switched affine systems}
\label{sec:r_tol_sw}

In this section, we propose an extension of the hyperplane method  to the case of a switched affine  system~of the form 
\begin{align}\label{eq:swi}
\dot x(t)=A_{\sigma(t)}x(t) + b_{\sigma(t)},
\end{align}
where the switching signal $\sigma(t)\in\Z[1,I]$    is the  input function, $I \ge 2$ is the number of modes,  $x(t)\in \mathbb{R}^n$ and $A_{i}\in \mathbb{R}^{n\times n}, b_{i} \in \mathbb{R}^n$ for all $i\in\Z[1,I]$. We make the following assumption.

\begin{assumption}\label{ass:input}
The   switching signal $\sigma(t)$  switches  $K$ times within 
the finite set $\Z[1,I]$ at fixed switching   instants  $0=t_0 <\ldots <t_{K+1}=T$,
 that is, $\sigma\in{\mathcal S}_{I}^K$, where
$$
{\mathcal S}_{I}^K:=\{\sigma\mid \sigma(t)= {i_k}\in \Z[1,I], \forall t\in[t_k  ,t_{k+1}),   k\in\N[0,K]\}.
$$
\end{assumption}

For every $k\in\N[0, K]$ and $i\in\N[1,I]$ we define $\bar A_{i}^k:=e^{A_{i}{(t_{k+1}-t_k)}}$ and \sloppy $\bar b_{i}^k=[\int_0^{(t_{k+1}-t_k)} e^{A_{i}\tau} d\tau]  b_{i}$. Moreover, we set $x_k:=x(t_k)$. 
 Note that under Assumption \ref{ass:input} the reachable set of system~\eqref{eq:swi} consists of a finite number of points that can be computed by solving  the state equations for each possible switching signal. Since the cardinality of the  set  ${\mathcal S}_{I}^K$ grows exponentially with $K$,  this approach is however computationally infeasible even for small systems.
We here show that, on the other hand,  the  hyperplane constants defined in \eqref{intercept}  can be  computed by solving a mixed integer linear program {(MILP)}, thus  allowing us to exploit the  sophisticated software  that has been developed to solve large MILPs in the last years.
\begin{proposition}[Tangent hyperplanes for switched  affine systems]\label{tangent2}
Consider system~\eqref{eq:swi} and suppose that Assumption~\ref{ass:input} holds.  Take a vector  
${\bf M}\in\R^n$ such that  ${\bf M}\ge |x_k|$ component-wise for all $k\in\N[0,K]$. Then 
\begin{align}
\intercept_T(c)=\max_{x_k,z_i^k,\gamma_i^{k}} &\quad c^\top x_{K+1}\label{MILP}\\
 \textup{s.t.}
 &\quad   z_{i}^{k+1}\le  (\bar A_{i}^{k}x_{k}+\bar b_{i}^{k})+{\bf M}(1-\gamma_i^{k}),\notag\\ 
 &\quad z_{i}^{k+1}\ge  (\bar A_{i}^{k}x_{k}+\bar b_{i}^{k})-{\bf M}(1-\gamma_i^{k}),\notag\\
 & \quad z_{i}^{k+1}\ge  -{\bf M}\gamma_i^{k},  \quad z_{i}^{k+1}\le  {\bf M}\gamma_i^{k},\notag\\
 &\quad z^k_i \in \mathbb{R}^n, \quad \forall k\in\Z[1,K+1], \forall  { i\in\N[1,I]},\notag \\
   &\quad \gamma^k_i \in \{0,1\}, \quad \forall k\in\Z[0,K], \forall { i\in\N[1,I]},\notag \\
   &\quad \textstyle  x_{k}=\sum_{i=1}^I z_i^k \in \mathbb{R}^n,\quad \forall k\in\Z[1,K+1],\notag\\
 &\quad   \textstyle \sum_{i=1}^I \gamma_i^{k}=1, \quad \forall k\in\Z[0,K],\notag\\
 &\quad x_0\in \mathbb{R}^n {\rm \ assigned.} \notag  \qedhere
\end{align}
\end{proposition}

\begin{proof}
To prove the statement we follow a  procedure similar to the one in \cite[Section~IV.A]{bemporad}.
Under Assumption~\ref{ass:input} the switching  signal $\sigma(t)$ is such that
$\sigma(t)=i_k, \forall t\in[t_k,t_{k+1}), \forall k\in\N[0,K]$. Therefore, the finite time optimal control  problem in \eqref{intercept} can  be rewritten as
\begin{align}
\intercept_T(c)=\max_{{i_k}\in \{1,\ldots, I\}} &\quad c^\top x_{K+1}\label{discrete}\\
 \mbox{s.t.}&\quad   x_{k+1}=\bar A_{i_k}^kx_k+\bar b_{i_k}^k\quad \forall k \in\Z[0,K]\notag\\
 &\quad x_0\in \mathbb{R} {\rm \ assigned.}  \notag
\end{align}
Let us introduce the binary variables $\gamma_i^{k}\in\{0,1\}$ defined so that, for each $ i\in\N[1,I]$ and $k \in\Z[0,K]$, $\gamma_i^{k} =1$ if and only if the system~is in mode $i$   in the time interval $[t_k,t_{k+1})$. Moreover, let us introduce a copy of the state vector for each possible update of the system~in each possible mode: 
$z_i^{k+1}=(\bar A_{i}^{k}x_{k}+\bar b_{i}^{k})\gamma_i^{k}$.
Then \eqref{discrete} is equivalent to the following optimisation problem
\begin{align}
\intercept_T(c):=\max_{x_k,z_i^k,\gamma_i^{k}} &\quad c^\top x_{K+1}\label{con}\\
 \mbox{s.t.}
 &\quad   z_i^{k+1}=(\bar A_{i}^{k}x_{k}+\bar b_{i}^{k})\gamma_i^{k}, \quad \forall i \in\Sigma,  \notag \\
 &\quad   \textstyle \sum_{i=1}^I \gamma_i^{k}=1, \quad \forall k \in\Z[0,K],\notag\\
&\quad \textstyle  x_{k}=\sum_{i=1}^I z_i^k,\quad \forall k \in\Z[1,K+1]\notag,\\
 &\quad x_0\in \mathbb{R} {\rm \ assigned.}  \notag
\end{align}
Finally, by using the big-M method    in \cite[Eq. (5b)]{bemporad1999control}, the first equality  constraint in the optimization problem \eqref{con} can be equivalently replaced by
\begin{align*}
z_{i}^{k+1}&\le  (\bar A_{i}^{k}x_{k}+\bar b_{i}^{k})+{\bf M}(1-\gamma_i^{k}), \qquad z_{i}^{k+1}\ge  -{\bf M}\gamma_i^{k},\\
z_{i}^{k+1}&\ge  (\bar A_{i}^{k}x_{k}+\bar b_{i}^{k})-{\bf M}(1-\gamma_i^{k}), \qquad z_{i}^{k+1}\le  {\bf M}\gamma_i^{k},
\end{align*}
leading to the equivalent reformulation given in \eqref{MILP}. \end{proof}\

We summarize our results on the hyperplane method for switched affine systems in the next corollary, which is an immediate consequence of Proposition \ref{tangent2} and Theorems \ref{thm:hm1}, \ref{thm:hm2}.

\begin{corollary}[The hyperplane method for switched affine systems] 
\label{thm:hm1_sw}
 Given system~\eqref{eq:swi},  let $x_0\in \mathbb{R}^n$ be the initial state  and suppose that Assumption \ref{ass:input} holds. Let $\intercept_T(c^d)$ be computed as in \eqref{MILP}.
 Then $\mathcal{R}^{out}_T(x_0)$ and $  \mathcal{R}^{y,out}_{T}(x_0)$  as defined in Theorems \ref{thm:hm1} and \ref{thm:hm2}  are outer approximations of $\mathcal{R}_T(x_0) $ and $ \mathcal{R}^{y}_{T}(x_0)$,  respectively. \hfill{$\square$}
 \end{corollary}

Note that in the case of switched affine systems it is not possible to recover an inner approximation, since there is no guarantee in general that the reachable set is convex. By computing the convex hull of the points $x_{K+1}$ in \eqref{MILP} for each direction $c$ one could however recover an inner approximation of the convex hull of $\mathcal{R}_T(x_0)$.

\section{Controlled stochastic biochemical reaction networks}
\label{CME}

A   biochemical reaction network is a system~comprising $S$ molecular species $Z_1$, ..., $Z_{S}$ that interact through $R$ reactions. Let $Z(t)=[Z_1(t), ...,Z_{S}(t)]^\top$  be the vector describing the number of molecules present in the network for each species at time $t$, 
 that is, the state of the network at time $t$. 
 Since each reaction $r$ is a stochastic event \cite{gillespie1992rigorous}, 
 $Z(t)$ is a stochastic process. In the following, we always use the upper case to denote a process and the lower case to denote its realizations.
 For example, $z=[z_1, ...,z_{S}]^\top$ denotes a particular realization of the state~$Z(t)$ of the stochastic process at time~$t$.
 
 A typical reaction $r\in  \mathbb{N}[1,R]$ can be expressed as
\begin{equation}\label{eq:reaction}
\nu'_{1r} Z_1 +\ldots+ \nu'_{Sr} Z_S\quad \longrightarrow \quad \nu''_{1r} Z_1 +\ldots+ \nu''_{Sr}Z_S,
\end{equation}
where $\nu'_{1r} ,\ldots, \nu'_{Sr} \in\N$ and $\nu''_{1r} ,\ldots, \nu''_{Sr}  \in\N$ are the coefficients that determine 
how many molecules for each species are respectively consumed and produced by the reaction.  
 The net effect of each reaction can thus be summarized with the \textit{stoichiometric vector} $\nu_r\in \N^S$,  whose components are $\nu''_{sr}-\nu'_{sr}$ for $s=1,\ldots,S$.  We say that a reaction is of order $k$ if it involves $k$ reactant units (i.e., $\sum_{s=1}^S \nu'_{sr}=k$) and we distinguish two classes of reactions: \\ 
 -\textit{uncontrolled} reactions  that  happen, in the infinitesimal interval $[t,t+dt]$, with probability 
\begin{equation}\label{eq:mass_action}
\alpha_r(\theta_r,z) dt:=\theta_r \cdot h_r(z) \cdot dt,
\end{equation}
where $h_r(z)$ is a given function of the available molecules $z$ and $\theta_r\in\R_{\ge 0}$ is the so-called rate parameter; \\ - \textit{controlled} reactions for which there exists an external signal $\sigmau_r(t)$ such that the reaction fires at time $t$ with probability
 \begin{equation}\label{eq:mass_action_control}
\sigmau_r(t) \cdot \alpha_r(\theta_r,z) dt.
\end{equation}
In the following we refer to $\alpha_r(\theta_r,z)$ as the propensity of the reaction and
without loss of generality we assume that the controlled reactions are the first $\Rc$ ones. If  $h_r(z) :=\Pi_{s=1}^S \binom{z_s}{\nu'_{sr}} $ we say that reaction $r$ follows the laws of {\em mass action kinetics} as derived in \cite{gillespie1992rigorous}. Our analysis  can however be applied to generic functions $h_r(z)$, allowing us  to model different types of kinetics, as the Michaelis-Menten \cite[Section~7.3]{wilkinson2011stochastic}. 

To illustrate the following results, we consider  a  model of gene expression as running example.
\begin{example}[Gene expression reaction network]\label{ex:gene}
Consider a biochemical  network consisting of two species,  the mRNA ($M$) and the corresponding protein ($P$), and the following reactions
\begin{align*}
\emptyset \quad & \xrightarrow{\hspace{0.3cm} \alpha_1(k_r,z)\hspace{0.3cm} } \quad M \hspace{1cm} M  \quad  \xrightarrow{\hspace{0.3cm}\alpha_3(k_p,z) \hspace{0.3cm}}  \quad M+P \\
M \ \quad & \xrightarrow{\hspace{0.3cm}\alpha_2(\gamma_r,z)\hspace{0.3cm}}  \quad \emptyset \hspace{1.2cm}  P \ \quad  \xrightarrow{\hspace{0.3cm}\alpha_4(\gamma_p,z) \hspace{0.3cm}}  \quad \emptyset
\end{align*}
where  the  parameters
$k_r$ and $k_p$ are the mRNA and protein production rates, while $\gamma_r$ and $\gamma_p$ are the mRNA and protein degradation rates, respectively. 
The empty set notation is used whenever a certain species is produced or degrades without involving the other species.
In this context,  $Z=[M,P]^\top$, $z=[m,p]^\top$, $\theta=[\theta_1,\theta_2,\theta_3,\theta_4]^\top :=  [k_r,\gamma_r,k_p,\gamma_p]^\top$ and the stoichiometric matrix is
 $$\nu:=[\nu_1, \nu_2,\nu_3 , \nu_4]=\left[\begin{array}{cccc}1 & -1 & 0 & 0 \\0 & 0 & 1 & -1\end{array}\right].$$
In the case of mass action kinetics the propensities $\alpha_r(\theta_r,z)$ can be further specified as
$
\alpha_1(k_r,z)=k_r, \ \alpha_2(\gamma_r,z)=\gamma_r\cdot m,\  \alpha_3(k_p,z)=k_p\cdot m, \ \alpha_4(\gamma_p,z)=\gamma_p\cdot p.$
\end{example}

Note that since the propensity of each reaction   depends only on the current state of the system, the process $Z(t)$ is Markovian. Let $p(t,z):=\mathbb{P}[Z(t)=z]$ be the probability that the realization of the process $Z$ at time $t$ is $z$.  Following the same procedure as in \cite{gillespie1992rigorous} one can derive a set of equations, known as {\em chemical master equation} (CME), describing the evolution of $p(z,t)$ as a function of the  external signal $u(t)$

{\small{ \begin{align}
&\dot p(z,t)=\sum_{r=1}^\Rc \left[p(z-\nu_r,t) \alpha_r(\theta_r, z-\nu_r)- p(z,t) \alpha_r(\theta_r, z)\right]\sigmau_r(t)\notag  \\&+ \!\!\!\sum_{r=\Rc+1}^R \left[p(z-\nu_r,t) \alpha_r(\theta_r, z-\nu_r) - p(z,t) \alpha_r(\theta_r, z)\right],\ \ \forall z\in\N^S.\label{eq:CME_control}
\end{align}}}
Since the previous set of equations depends on the external signal $u$ we refer to it as the  \textit{controlled CME}.
Typical biochemical reaction networks involve many different species, whose counts can theoretically grow unbounded. Consequently,  the controlled {CME}  in  \eqref{eq:CME_control} is a system~of infinitely many  coupled ordinary differential equations that cannot be solved, even for very simple systems. Several analytical and computational  methods have been proposed in the literature to circumvent this difficulty, see \cite{wilkinson2011stochastic, goutsias2013markovian,ruess2014moment}  for a comprehensive review. In the following we limit our discussion to  two  methods:  \textit{moment equations} \cite{hespanha2008moment} and  \textit{finite state projection}  (FSP)  \cite{fsp}.

\subsection{The moment equations}
\label{sec:th:mom}

We start by considering  the case when all the reactions follow the laws of mass action kinetics and are at most of order one. In this case for each reaction $r$
the propensity $h_r(z)$ is  affine in the molecule counts vector $z$ and one can show that the moments equations are closed (i.e., the dynamics of moments up to any order $k$ do not depend on higher order moments), see for example \cite{lee2009moment}. 
Let $x_{\le 2}(t)$ be a vector whose components are the moments of $Z(t)$ up to second order. From \cite[Equations (6) and (7)]{lee2009moment} one gets
\begin{equation}\label{eq:controlled_moments_affine}
\dot x_{\le 2}(t)=A(\sigmau(t))x_{\le 2}(t)+b(\sigmau(t)).
\end{equation}

\begin{example}\label{ex:gene1}
Consider the  gene expression model of Example~\ref{ex:gene}. Assume that the reactions follow the mass action kinetics and that an external input signal influencing the first reaction, that is the mRNA production, is available (as in
 \cite{milias2011silico, parise2015guides,Olson2014a,Batt2012,menolascina2011analysis}), so that 
 $\alpha_1(k_r,z):=k_r\cdot \sigmau(t)$. Set 
$$x_{\le 2} :=[\mathbb{E}[M], \mathbb{E}[P], \mathbb{V}[M], \Cov[M,P], \mathbb{V}[P] ]^\top.$$
Then the moments evolution over time is expressed as
\bev 
\dot x_{\le 2}(t)=A x_{\le 2}(t)+B  \sigmau(t),
\label{eq:gene_1}
\eev 
where 
{\small \begin{align*}
A&=\left[ \arraycolsep=1.4pt\begin{matrix}- \gamma_r  & 0 & 0 & 0 & 0 \cr
k_p & -\gamma_p & 0 & 0 & 0\cr
 \gamma_r  & 0 & - 2  \gamma_r   & 0 & 0 \cr
0 & 0 & k_p & - ( \gamma_r  +\gamma_p) & 0\cr
k_p & \gamma_p & 0 & 2 k_p & - 2 \gamma_p\end{matrix}\right], \quad 
B=\ \left[\begin{matrix}
k_r  \cr  0 \cr k_r   \cr 0  \cr 0  \end{matrix} \right].
\end{align*}}
\end{example}

 Since the input $\sigmau(t)$ may appear  in the entries of the $A$ matrix, the moment equations~\eqref{eq:controlled_moments_affine} are in general  nonlinear. To overcome this issue we introduce the following assumption on the external signal $\sigmau(t)$.

\begin{assumption}\label{ass:input_switch}
The  external signal $\sigmau(t)$ can switch at most  $K$ times within the set $\Sigma^d$, as defined in Assumption \ref{ass:linear_input}, at preassigned switching   instants  $0=t_0 <\ldots <t_{K+1}=T$.
\end{assumption}

Assumption \ref{ass:input_switch} imposes that  the number of switchings and their timing during a given experiment is fixed a priori. This assumption can be motivated by the fact that changes in the external stimulus are costly and/or stressful for the cells. Moreover, it is trivially satisfied if the stimulus can only be changed simultaneously with some fixed events, such as culture dilution or measurements. The great advantage of Assumption \ref{ass:input_switch} is that, as illustrated in the following remark, it allows us to rewrite the nonlinear moment equations \eqref{eq:controlled_moments_affine} as a switched affine system so that the theoretical tools described in Section~\ref{sec:r_tol_sw} can be applied.

\begin{remark}
\label{construction}
The set $\Sigma^d$ has finite cardinality $I:=\Pi_{r=1}^m q_r$ and we can enumerate its elements as $u^i, i\in\N[1,I]$. Consequently, for any fixed external signal $\sigmau(t)$ satisfying Assumption \ref{ass:input_switch} we can construct a sequence of indices 
in $\N[1,I]$ such that, at any time $t$, $\sigma(t)=i$ if and only if $\sigmau(t)=u^i$. Such switching sequence  $\sigma$ satisfies Assumption \ref{ass:input}.
\end{remark}

\subsection{The finite state projection}
\label{sec:th:fsp}

Let us introduce a  total ordering $\{z^j\}_{j=1}^\infty$   in the set of all possible state realizations $z\in\mathbb{N}^S$. For the system~in Example~\ref{ex:gene}, we could for instance use the mapping
\begin{align*}
z^1&=(0,0), \ z^2=(1,0),\  z^3=(0,1),\ z^4=(2,0), \\
z^5&=(1,1), \ z^6=(0,2),\  z^7=(3,0),\ z^8=(2,1),\ \ldots
\end{align*}
where $(m,p)$ denotes the state with $m$ mRNA copies and $p$ proteins (see Fig. \ref{fig:FSP}).
\begin{figure}[H]
\begin{center}
\includegraphics[width=0.28\textwidth]{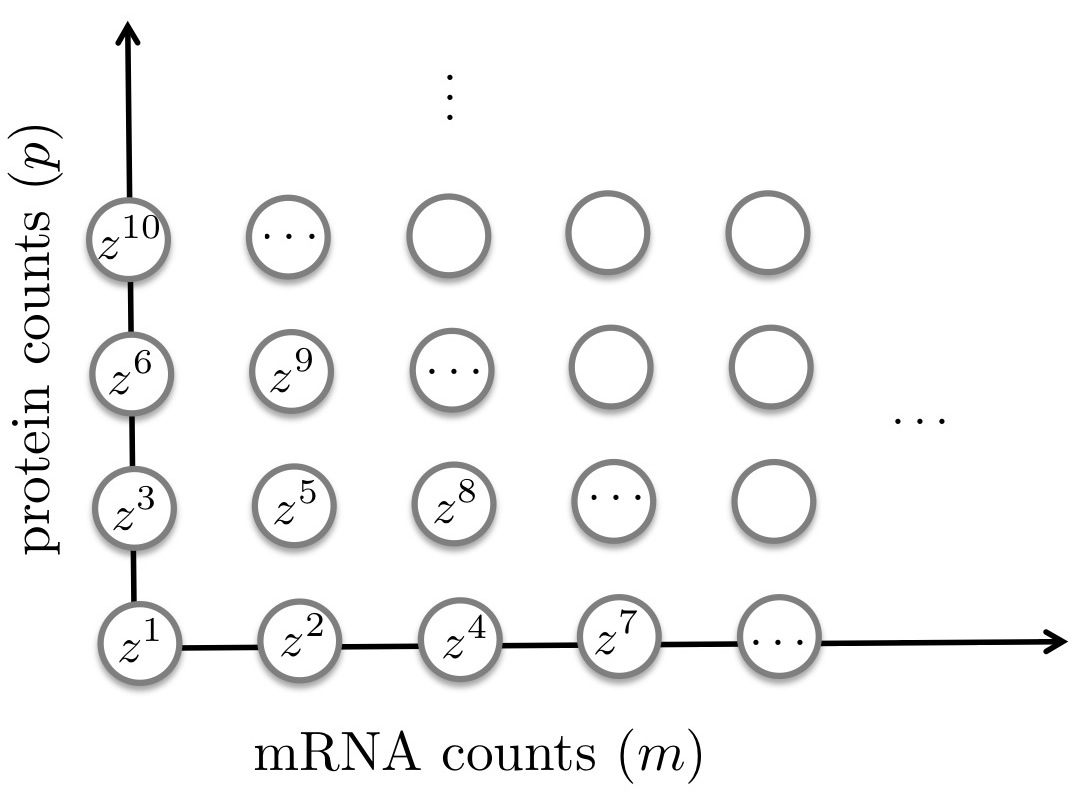}
\end{center}
\vspace{-0.5cm}
\caption{State space for the gene expression system~of Example \ref{ex:gene}.}
\label{fig:FSP}
\end{figure}
Following the same steps as in  \cite{fsp} and setting\footnote{Not to be confused with the symbol used to denote  the amount of protein.}   $P_j(t):=p(z^j,t)$, the \textit{controlled} {CME}  in \eqref{eq:CME_control} can be rewritten as the nonlinear infinite dimensional system~
\begin{align}
\dot P(t)=F({\sigmau(t)})P(t),
\label{eq:lin_inf_s}
\end{align}
where  $P(t)$ is an infinite dimensional vector with entries in $[0,1]$. If the   signal

$\sigmau(t)$ satisfies Assumption \ref{ass:input_switch}, then \eqref{eq:lin_inf_s} can be rewritten as an infinite dimensional linear switched system
\begin{align}
\dot P(t)=F_{\sigma(t)}P(t),
\label{eq:lin_inf}
\end{align}
with switching signal $\sigma(t)$ constructed from $\sigmau(t)$ as detailed in Remark \ref{construction},  $I=\Pi_{r=1}^m q_r$ modes and matrices $F_i:=F({\sigmau^i})$.   Note that system~\eqref{eq:lin_inf} can also be thought of as a Markov chain with countably many states $z^j\in\mathbb{N}^S$ and time-varying transition matrix $F_{\sigma(t)}$. 

As in  the FSP method  for the uncontrolled {CME} \cite{fsp},  one can try to approximate the behavior of the infinite Markov chain in \eqref{eq:lin_inf} by constructing a reduced Markov chain that keeps track of the probability of visiting only the states  indexed in a suitable set $J$. To this end, let us define the reduced order system
\begin{align}
\dot {\bar P}_J(t)=\left[F_{\sigma(t)}\right]_J\bar P_J(t), \quad \bar P_J(0)= P_J(0), 
\label{eq:lin_fin}
\end{align}
where $P_J(0)$ is the subvector of  $P(0)$ corresponding to the indices in $J$, and  $[F]_J$ denotes the submatrix of $F$ obtained by selecting only the rows and columns with indices in $J$. 
 Note that while the full matrix $F_{\sigma(t)}$ is stochastic, the reduced matrix $\left[F_{\sigma(t)}\right]_J$  is substochastic. Consequently,  the probability mass is in general not  preserved in \eqref{eq:lin_fin} (i.e. $\mathbbm{1}^\top  {\bar P}_J(t)$ may decrease with time).
From now on, we denote by $P(T;\sigma)$ and $\bar P_J(T;\sigma)$ 
 the solutions at time $T$ of system~\eqref{eq:lin_inf} and system~\eqref{eq:lin_fin},  respectively, when the switching signal $\sigma$ is applied. The dependence on the initial conditions  $P(0)$ and $P_J(0)$ is omitted to keep the notation  compact. As in the uncontrolled case, the truncated system~\eqref{eq:lin_fin} is a good approximation of the original system~\eqref{eq:lin_inf}  if most of the probability mass lies in $J$. However in the controlled case we need to guarantee that this happens for all possible switching signals. This intuition  is formalized  in the following assumption.
\begin{assumption}
\label{a19}
For a given finite set of state indices $J$, an  initial condition $P_J(0)$, a given tolerance $\varepsilon>0$ and a finite instant $T>0$, 
\begin{equation}
\mathbbm{1}^\top \bar P_J(T;\sigma)\ge 1-\varepsilon,\quad  \forall \sigma\in  {\mathcal S}_{I}^K.
\end{equation}
\end{assumption}

\noindent Note that Assumption \ref{a19} holds if and only if 
\begin{align*}
1-\varepsilon \le \min_{\sigma\in  {\mathcal S}_{I}^K} &\quad \mathbbm{1}^\top \bar P_J(T;\sigma)\\
 \mbox{s.t.}&\  \dot {\bar P}_J(t;\sigma)\!=\!\left[F_{\sigma(t)}\right]_J\bar P_J(t;\sigma), \  \bar P_J(0)= P_J(0).
 \end{align*}
This problem has the same structure as \eqref{intercept}. Therefore, as   illustrated in Section~\ref{sec:r_tol_sw}, Assumption \ref{a19} can be checked by solving the MILP \eqref{MILP}  for the switched affine system \eqref{eq:lin_fin}  by setting $c=\mathbbm{1}$ and ${\bf M}=\mathbbm{1}$.  Under Assumption \ref{a19}, the following relation between the solutions of \eqref{eq:lin_inf} and \eqref{eq:lin_fin} holds.

\begin{proposition}[FSP for controlled CME] \label{fsp}
If Assumptions \ref{ass:input} and \ref{a19} hold, then for every switching  signal $\sigma\in \mathcal S^K_I $, it holds
\begin{align*}
&P_j(T;\sigma)\ge\bar P_j(T;\sigma),\quad  \forall j\in J \\
&\|P_J(T;\sigma)-\bar P_J(T;\sigma)\|_1\le \varepsilon. \qedhere
\end{align*}
\end{proposition}

\begin{proof}
This result has  been proven in \cite{fsp} for linear systems. We  extend it here to the case of switched systems with $K$ switchings. Note that for any $i\in\N[1,I]$, $F_i:=F({u^i})$ has non-negative off diagonal elements \cite{fsp}. Hence, using the same argument as in \cite[Theorem 2.1]{fsp} it can be shown that for any index set $J$, and any $\tau\ge0$
$$[{\rm exp}(F_i\tau)]_{J}\ge {\rm exp}([F_i]_{J}\tau)\ge 0, \quad \forall i \in 1,\ldots,I .$$
Consider  an arbitrary switching signal $\sigma\in \mathcal S^K_I$. We  have
\begin{align}
P_J(T;\sigma)&=[\Pi_{k=0}^{K}{\rm exp} (F_{i_k}(t_{k+1}-t_k))\cdot P(0)]_{J}\label{p1}  \\&\ge \Pi_{k=0}^{K}[{\rm exp}(F_{i_k}(t_{k+1}-t_k))]_J\cdot P_J(0)\notag\\
&\ge \Pi_{k=0}^{K}{\rm exp} ([F_{i_k}]_J(t_{k+1}-t_k))\cdot P_J(0)=\bar P_J(T;\sigma).\notag
\end{align}
Moreover, from $1=\sum_{j=1}^\infty P_j(T;\sigma)\ge \sum_{j\in J} P_j(T;\sigma)=\mathbbm{1}^\top  P_J(T;\sigma)$ and Assumption~\ref{a19}, we get
\begin{align}
\mathbbm{1}^\top \bar P_J(T;\sigma)\ge 1-\varepsilon \ge \mathbbm{1}^\top  P_J(T;\sigma) -\varepsilon \label{p2}.
\end{align}
Combining \eqref{p1} and \eqref{p2} yields
$ 0\le \mathbbm{1}^\top  P_J(T;\sigma) - \mathbbm{1}^\top \bar P_J(T;\sigma) \le \varepsilon$, thus $\|P_J(T;\sigma)-\bar P_J(T;\sigma)\|_1\le \varepsilon$.
\end{proof}

\section{Analysis of the reachable set} 
\label{sec:reach}

We here show how the reachability tools of Sections~\ref{sec:r_tol} and~\ref{computing} can be applied to the moment equation and FSP reformulations derived in Sections \ref{sec:th:mom} and \ref{sec:th:fsp},  under different assumptions. Fig. \ref{fig:ch9} presents a conceptual scheme of this section.
\begin{figure}[H]
\begin{center}
\includegraphics[width=0.4\textwidth]{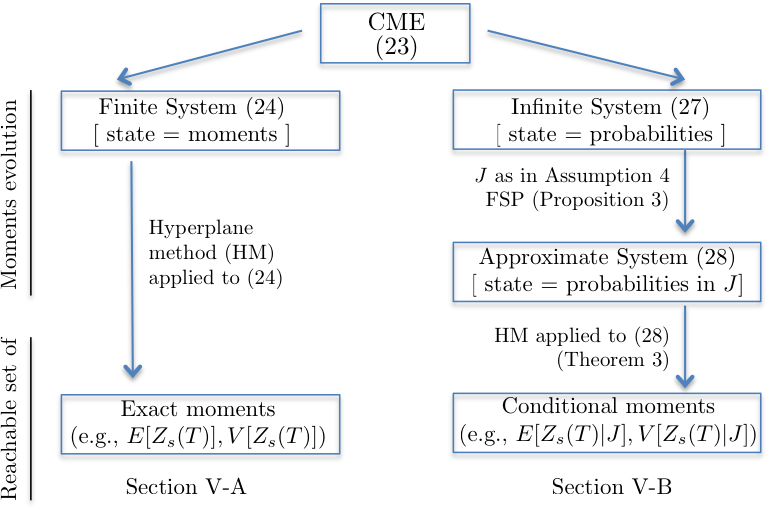}
\end{center}
\vspace{-0.5cm}
\caption{\scriptsize Conceptual scheme for the reachable set analysis of biochemical networks.}
\label{fig:ch9}
\end{figure}

\subsection{Reachable set of networks with affine propensities via moment equations}
\label{sec:r_affine}

The methods developed in   Sections~\ref{sec:r_tol} and \ref{computing}  can be  applied to the moments equations in \eqref{eq:controlled_moments_affine} to approximate the desired projected reachable set. To illustrate the proposed procedure, we distinguish two cases depending on whether the external signal $u(t)$  influences reactions of order zero or one.\

\subsubsection{Linear moments equations}
\label{sec:r_affine_af}

We start by considering the case when all and only the reactions of order zero are controlled, so that $h_r(z)=1 $ for $r\in\Z[1,\Rc]$ and  $h_r(z)={\nu_r'}^\top z$ for $r\in\Z[\Rc+1,R]$. This is  the simplest scenario since the system~of moment equations given in \eqref{eq:controlled_moments_affine} becomes linear
\begin{equation}\label{eq:controlled_moments_linear}
\dot x_{\le 2}(t)=Ax_{\le 2}(t)+Bu(t),
\end{equation}
see \cite[Equations (6) and (7)]{lee2009moment}.
 Consequently, the theoretical results of Section~\ref{sec:r_tol_af} can be  applied to \eqref{eq:controlled_moments_linear} by setting $\sigma(t)\equiv \sigmau(t)$.
 If the external signal $u\equiv\sigma$ satisfies Assumption \ref{ass:linear_input}, both inner and  outer approximations of the reachable set can be computed by using Corollary \ref{cor:linear}.

\subsubsection{Switched affine moments equations}
\label{sec:r_affine_sw}

If  reactions of order one are controlled then the external input $\sigmau(t)$ appears also in the entries of the $A$ matrix and system~\eqref{eq:controlled_moments_affine} is nonlinear. To overcome this issue we exploit Assumption \ref{ass:input_switch}. 
Specifically, let $\sigma(t)$ be the switching signal associated with $u(t)$ as described in Remark \ref{construction}. Then 
\eqref{eq:controlled_moments_affine} can be equivalently rewritten as the switched affine system
\begin{equation}\label{eq:controlled_moments_sww}
\dot x_{\le 2}(t)=A_{\sigma(t)}x_{\le 2}(t)+b_{\sigma(t)},
\end{equation}
 with matrices
$A_i:=A(\sigmau^i),  b_i:=b(\sigmau^i)$, for all $i\in\N[1,I]$.
 Consequently, the theoretical results of Section~\ref{sec:r_tol_sw} can be  applied to \eqref{eq:controlled_moments_sww} and an outer approximation of the reachable set can  be computed by using Corollary~\ref{thm:hm1_sw}.

\subsection{Reachable set of networks with generic propensities via finite state projection}
\label{sec:r_gen}

 If the network contains reactions of order higher than one or if the reactions do not follow the laws of mass action kinetics, then $h_r(z)$ might be non-affine. In such cases, the arguments illustrated in the previous subsection cannot be applied. We here show how  the FSP approximation of the {CME}  derived in Section \ref{sec:th:fsp} 
 can be used to  overcome this problem. 

 Firstly note that, from system~\eqref{eq:lin_inf}, one can compute the evolution of the \textit{uncentered} moments of $Z(t)$, as a linear function of $P(t)$.
  \footnote{The reachable set for the \textit{centered} moments can be immediately computed from the reachable set of the \textit{uncentered} ones, since there is a bijective relation between the set of  \textit{centered} and \textit{uncentered}  moments up to any desired order.}
For example, if we let $z_s^j$ be the amount of species $Z_s$ in the state $z^j$, then  the mean $\mathbb{E}[Z_s]$ of any species $s$  can be obtained as $l^\top   P(t)$,  by setting
$
l:=\left[z_s^1,\ z_s^2,\ \hdots \right]^\top$,
and the second  uncentered moment $\mathbb{E}[Z^2_s]$ can be obtained as $l^\top   P(t)$, by setting
$l:=\left[(z_s^1)^2,\ (z_s^2)^2,\ \hdots \right]^\top.$
Consequently the desired projected reachable set coincides with the output reachable set of the infinite dimensional linear  switched system~\eqref{eq:lin_inf} with linear output 
\begin{align}\label{cm}
y(t)=\begin{bmatrix} l^{1},\ \  l^{2} \end{bmatrix}^\top P(t),
\end{align}
where $l^1$ and $l^2$ are the infinite vectors associated with any desired pair of moments. Note that $l^1$ and $l^2$ are non-negative.

\smallskip 

\textbf{Example \ref{ex:gene} (cont.)}\textit{
With the ordering introduced at the beginning of the section, the uncentered protein moments up to order two  can be computed as the output of \eqref{eq:lin_inf} by setting
\begin{equation}
\begin{aligned}
l^1&=\left[\begin{array}{ccccccccc}0 & 0 & 1& 0&1 & 2 & 0& 1& \ldots\ \end{array}\right]^\top,\\
\hspace{1cm} l^2&=\left[\begin{array}{ccccccccc}0 & 0 & 1& 0&1 & 4 & 0& 1& \ldots \ \end{array}\right]^\top. 
\end{aligned}
\end{equation}}

 Let $l_j^{1}$ and $l_j^{2} $ be the $j$-th components of the vectors $l^{1}$ and $l^{2}$, respectively, as defined in \eqref{cm}. For a given species of interest $s$ and set $J$, we denote by
{\small \begin{equation}
y_1(t;\sigma)\!:=\!\frac{\sum_{j\in J} l_j^{1}\cdot P_j(t;\sigma)}{\sum_{j\in J} P_j(t;\sigma)}, \
y_2(t;\sigma)\!:=\!\frac{\sum_{j\in J} l_j^{2}\cdot P_j(t;\sigma)}{\sum_{j\in J} P_j(t;\sigma)}  \label{output_conditional}
\end{equation}}
the moments associated with $l^1$ and $l^2$ \textit{conditioned on the fact that $Z(t)$ is in $J$ and the switching  signal  $\sigma$ is applied}. For example if one is interested in the mean and second order moment of a specific species  $Z_s(t)$ we get $y_1(t;\sigma)=\mathbb{E}\left[Z_s(t) \mid Z(t)\in J, \sigma(\cdot)\right] $ and $y_2(t;\sigma)=\mathbb{E}\left[Z^2_s(t) \mid Z(t)\in J,\sigma(\cdot) \right]$.
The aim of this section is to obtain an outer approximation of the output reachable set of the infinite system~\eqref{eq:lin_inf} with the nonlinear output \eqref{output_conditional}, by using  computations involving only the finite dimensional system~ \eqref{eq:lin_fin}. 
To this end, we define the two entries of the  output of the finite dimensional system~as
\begin{equation}
\begin{array}{l}
\bar y_1(t;\sigma):={\sum_{j\in J} l_j^{1}\cdot \bar P_j(t;\sigma)}=:(\bar l^1)^\top \bar P_J(t;\sigma) \cr
\bar y_2(t;\sigma):={\sum_{j\in J} l_j^{2}\cdot \bar P_j(t;\sigma)}=:(\bar l^2)^\top \bar P_J(t;\sigma).  
\end{array}\label{output_finite}
\end{equation}
\begin{theorem}
\label{krogh_inf}
Suppose Assumptions \ref{ass:input_switch} and \ref{a19} hold. Let ${\mathcal{R}}_T^y(x_0)$  be the output reachable set at time $T>0$ of system~\eqref{eq:lin_inf} with output \eqref{output_conditional}.
Choose $D$ values $\gamma^d\in\mathbb{R}$ and set  $c^d:=(\bar l^2)-\gamma^d (\bar l^1) \in\mathbb{R}^n$, with $\bar l^1,\bar l^2$  as in \eqref{output_finite}.
Set 
\begin{align*}
\mathcal{H}^{y}_{T}(\gamma^d)&:=\{ w \in\mathbb{R}^2 \mid w_2 \le \gamma^d w_1+ \bar \intercept_T(c^d) +\delta(\gamma^d)\},
\end{align*}
where $\bar \intercept_T(c^d)$ is the constant that makes the hyperplane  
$\textup{H}_T(c^d)$  in \eqref{eq:half_space} tangent to the reachable set of the finite  system~ \eqref{eq:lin_fin} (i.e. $\bar \intercept_T(c^d)$ can be computed as in \eqref{MILP})
and 
$$\textstyle \delta(\gamma^d):=\frac{2\varepsilon }{1-\varepsilon}\cdot (\max\{0,-\gamma^d\}\cdot \|\bar l^1\|_\infty+\|\bar l^2\|_\infty),$$
with $\varepsilon$ as in  Assumption \ref{a19}. Then the set 
$
\mathcal{R}_{T}^{y,out}(x_0):=\cap_{d=1}^D\{\mathcal{H}_{T}^{y}(\gamma^d)\} 
$
is an outer approximation of $\mathcal{R}_{T}^y(x_0)$. 
\end{theorem}

\begin{proof}
Firstly note that if the external signal $u$ satisfies Assumption~\ref{ass:input_switch}  then the corresponding switching signal $\sigma(t)$ (constructed as in Remark \ref{construction}) satisfies Assumption \ref{ass:input}. Let $\bar{\mathcal{R}}_{T}^{y}(x_0)$ be the output reachable set of the finite dimensional system~\eqref{eq:lin_fin}  with output \eqref{output_finite}.  Proposition \ref{tangent2} guarantees that for any direction $c^d$  the constant  $\bar \intercept_T(c^d)$ that makes 
\begin{align*}
\bar{\mathcal{H}}^{y}_{T}(\gamma^d)&:=\{ w \in\mathbb{R}^2 \mid w_2 \le \gamma^d w_1+ \bar \intercept_T(c^d) \}
\end{align*}
tangent to $\bar{\mathcal{R}}_{T}^{y}(x_0)$ can be computed by solving the MILP~\eqref{MILP} for system \eqref{eq:lin_fin}. The main idea of the proof is to show that if we  shift  the halfspace $\bar{\mathcal{H}}^{y}_{T}(\gamma^d)$ by a suitably defined constant $\delta(\gamma^d)$ we can guarantee that the original reachable set ${\mathcal{R}}_{T}^{y}(x_0)$ is a subset of the shifted halfspace ${\mathcal{H}}^{y}_{T}(\gamma^d)$  defined in the statement. The result then follows  since $\mathcal{R}_{T}^{y,out}(x_0)$ is defined as the intersection of hyperspaces containing $\mathcal{R}_{T}^{y}(x_0)$.

To derive the constant $\delta(\gamma^d)$ we start by  focusing on the first component of the output  and for simplicity we will omit the dependence on $(T;\sigma)$ in  $P_j,\bar P_j,y$ and $\bar y$.
Take  any switching   signal $\sigma \in  \mathcal S^K_I $.    By taking into account the following conditions: (1) $l_j^{1}\ge0$ for all $j\in J$; (2) $P_j\ge \bar P_j$ for all $j\in J$, due to Proposition \ref{fsp}, and (3) $\sum_{j\in J} P_j \le 1$, we get $y_1\ge \bar y_1$. Consequently, at time $t=T$ we have
{
\begin{align*}
|y_1-\bar y_1|&=y_1-\bar y_1 =\textstyle\frac{\sum_{j\in J} l_j^{1}\cdot P_j}{\sum_{j\in J} P_j}- {\sum_{j\in J} l_j^{1}\cdot \bar P_j} \\& \le \textstyle  \frac{\sum_{j\in J} l_j^{1}\cdot P_j}{1-\varepsilon}- {\sum_{j\in J} l_j^{1}\cdot \bar P_j}\\
&=\textstyle \left(1+\frac{\varepsilon}{1-\varepsilon}\right) \sum_{j\in J} l_j^{1}\cdot P_j- {\sum_{j\in J} l_j^{1}\cdot \bar P_j}\\
&= \textstyle \frac{\varepsilon}{1-\varepsilon}\sum_{j\in J} l_j^{1}\cdot P_j + {\sum_{j\in J} l_j^{1}\cdot (P_j-\bar P_j)}\\
&\textstyle \le  \|\bar l^1\|_\infty \left(\frac{\varepsilon}{1-\varepsilon}\sum_{j\in J}  P_j + {\sum_{j\in J} (P_j-\bar P_j)}\right)\\
&\textstyle \le  \|\bar l^1\|_\infty \left(\frac{\varepsilon}{1-\varepsilon}+ \|P_J-\bar P_J\|_1\right) \textstyle  \le   \|\bar l^1\|_\infty\frac{2\varepsilon }{1-\varepsilon}, 
\end{align*}
where we used $\sum_{j\in J} P_j\ge \sum_{j\in J}\bar  P_j\ge 1-\varepsilon$ (due to Assumption \ref{a19}), and $P_j\ge \bar P_j, \|P_J-\bar P_J\|_1\le \varepsilon$ (following from Proposition \ref{fsp}).}
 To summarize,
$
\bar y_1\le y_1\le \bar y_1+  \|\bar l^1\|_\infty\frac{2\varepsilon }{1-\varepsilon}. 
$
Similarly, it can be proven that
$
\bar y_2\le y_2\le \bar y_2+  \|\bar l^2\|_\infty\frac{2\varepsilon }{1-\varepsilon}. 
$
Consider any pair $( y_1, y_2)\in {\mathcal{R}}_{T}^{y}(x_0)$ and the associated pair  $(\bar y_1,\bar y_2)\in \bar{\mathcal{R}}_{T}^{y}(x_0)$ (i.e. the two output pairs obtained from \eqref{eq:lin_inf} and \eqref{eq:lin_fin} when the same $\sigma$ is applied). 
Note that $(\bar y_1,\bar y_2)\in \bar{\mathcal{R}}_{T}^{y}(x_0)$ implies $(\bar y_1, \bar y_2)\in \bar{\mathcal{H}}^{y}_{T}(\gamma^d)$ for any  $\gamma^d.$
The previous relations  then  imply that if $\gamma^d\ge0$, 
\begin{align*}
y_2&\textstyle \le \bar y_2+  \|\bar l^2\|_\infty\frac{2\varepsilon }{1-\varepsilon} \le \gamma^d\bar y_1+ \bar \intercept_T(c^d) +  \|\bar l^2\|_\infty\frac{2\varepsilon }{1-\varepsilon} \\ &\textstyle \le \gamma^d y_1+ \bar \intercept_T(c^d) +  \|\bar l^2\|_\infty\frac{2\varepsilon }{1-\varepsilon} =  \gamma^d y_1+ \bar \intercept_T(c^d)+\delta(\gamma^d).
\end{align*} 
On the other hand, when $\gamma^d<0$ 
\begin{align*}
y_2&\textstyle\le \bar y_2+  \|\bar l^2\|_\infty\frac{2\varepsilon }{1-\varepsilon} \le \gamma^d\bar y_1+ \bar \intercept_T(c^d) +  \|\bar l^2\|_\infty\frac{2\varepsilon }{1-\varepsilon} \textstyle\\&\le \textstyle \gamma^d y_1+ \bar \intercept_T(c^d) + ( \|\bar l^2\|_\infty-\gamma^d  \|\bar l^1\|_\infty)\frac{2\varepsilon }{1-\varepsilon}   \\&=  \gamma^d y_1+ \bar \intercept_T(c^d)+\delta(\gamma^d).
\end{align*} 
Therefore for every  signal $\sigma$ and every $\gamma^d$ it holds
$
y_2(T;\sigma)\le  \gamma^d y_1(T;\sigma)+ \bar \intercept_T(c^d)+\delta(\gamma^d)
$
and consequently $[y_1(T;\sigma), y_2(T;\sigma)]^\top \in {\mathcal{H}}^{y}_{T}(\gamma^d)$.
\end{proof}

\section{ Analysis of single cell realizations}
\label{sec:individual}
The previous analysis   focused on characterising what combinations of moments of the stochastic biochemical reaction network are achievable by using the available external input. In this section, we change perspective and instead of looking at population properties  we focus on single cell trajectories. Specifically, we are interested in characterising the probability that  a single realization of the stochastic process will satisfy a specific property at the final time $T$ (e.g. the number of copies of a certain species  is higher/lower than a certain threshold) when starting from an initial condition $P(0)$. Note that we can start either deterministically from a given state $z^i$ (by setting $P(0)=e_i$) or  stochastically from any state according to a generic vector of probabilities $P(0)$. To  define the problem let us call $\mathcal{T}$ the target set, that is, the set of all indices $i$ associated with a state $z^i$ in the Markov chain \eqref{eq:lin_inf_s} that satisfies the desired property. Note that this set might be of infinite size. We restrict our analysis to external signals satisfying Assumption \ref{ass:input_switch}, so that we can map the external signal $u$ to the switching  signal $\sigma$, as detailed in Remark~\ref{construction}. For a fixed signal $\sigma$ the solution of \eqref{eq:lin_inf} immediately allows one to compute the probability that the state at time $T$ belongs to $\mathcal{T}$, and thus has the desired property, as 
$ \mathcal{P}_\mathcal{T}(\sigma):= \mathbbm{1}_\mathcal{T}^\top P(T;\sigma) $
where $\mathbbm{1}_\mathcal{T}$ is an infinite vector that has the $i$th component equal to $1$ if $i\in\mathcal{T}$ and $0$ otherwise. Our objective is to select the switching  signal $\sigma(t)$ (and thus the external signal $u(t)$) that maximizes the probability $ \mathcal{P}_\mathcal{T}(\sigma)$.\footnote{Note that  one can use the same tools to maximize the probability of  \textit{avoiding} a given set $\mathcal{D}$ by maximizing the probability of being in $\mathcal{T}=\mathcal{D}^c$.} That is, we aim at solving
\begin{align}\label{max_prob}
\mathcal{P}_\mathcal{T}^\star: = \max_{\sigma\in\mathcal{S}^K_I}\ \mathcal{P}_\mathcal{T}(\sigma), \qquad \sigma^\star:=  \arg\max_{\sigma\in\mathcal{S}^K_I}\ \mathcal{P}_\mathcal{T}(\sigma),
\end{align}
where $I$ is the cardinality of $\Sigma^d$ as by Remark~\ref{construction}. Note that  $\mathcal{P}_\mathcal{T}(\sigma)$ in \eqref{max_prob} is computed according to $P(T;\sigma)$ which is an infinite dimentional vector. In the next theorem we show how to overcome this issue and approximately solve  \eqref{max_prob} by using the FSP approach of Proposition \ref{fsp} and the reformulation as MILP given in Proposition~\ref{tangent2}. 
To this end, let
\begin{align}\label{max_finite}
 \bar{\sigma}^\star: =\arg\max_{\sigma\in\mathcal{S}^K_I} \bar{\mathcal{P}}_\mathcal{T}(\sigma).
\end{align}
where $\bar{\mathcal{P}}_\mathcal{T}(\sigma):= \bar{\mathbbm{1}}_\mathcal{T}^\top \bar P_J(T;\sigma)$ is the probability that the final state of the \textit{reduced} Markov chain
\eqref{eq:lin_fin}  belongs to $\mathcal{T}\cap J$ at time $T$ given the switching signal $\sigma$, and
$\bar{\mathbbm{1}}_\mathcal{T}$ is a vector of size $|J|$ that has $1$    in the positions  corresponding to states of $J$ that belong also to $\mathcal{T}$, and $0$ otherwise.
\begin{theorem}
Suppose that Assumptions \ref{ass:input_switch} and \ref{a19} hold. Then 
$$\mathcal{P}_\mathcal{T}( \bar{\sigma}^\star ) \ge \mathcal{P}_\mathcal{T}^\star - 2 \varepsilon.$$
Moreover \eqref{max_finite} can be solved by solving the MILP in \eqref{MILP} for  system \eqref{eq:lin_fin} with $c=\bar{\mathbbm{1}}_\mathcal{T}$ and $\bf{M}=\mathbbm{1}$.
\end{theorem}
\begin{proof}
Under Assumption \ref{ass:input_switch} and \ref{a19}, for any set $\mathcal{T}$ and any  signal $\sigma$, we get
\begin{align*}\mathbbm{1}_\mathcal{T}^\top P&\textstyle = \sum_{i\in\mathcal{T}}  P_i \le\!   \sum_{i\in\mathcal{T}\cap J}   P_i 
 + \sum_{i\notin  J}   P_i  \le   \sum_{i\in\mathcal{T}\cap J}   P_i  \!+\! \varepsilon \\
 &\textstyle \le   \sum_{i\in\mathcal{T}\cap J}   \bar{P}_i + \sum_{i\in\mathcal{T}\cap J}  | P_i -\bar{P}_i|
 + \varepsilon \\
 &\textstyle \le   \sum_{i\in\mathcal{T}\cap J}   \bar{P}_i + \| P_J -\bar{P}_J\|_1+\varepsilon =  \bar{\mathbbm{1}}_\mathcal{T}^\top    \bar{P} + 2 \varepsilon,
 \end{align*}
and
$
\mathbbm{1}_\mathcal{T}^\top P= \sum_{i\in\mathcal{T}}  P_i  \ge  \sum_{i\in\mathcal{T}\cap J}  P_i   \ge  \sum_{i\in\mathcal{T}\cap J}  \bar{P}_i =  \bar{\mathbbm{1}}_\mathcal{T}^\top    \bar{P},
$
 where we used Assumption \ref{a19} and Proposition \ref{fsp} and we omitted $(T;\sigma)$ for simplicity.  To sum up, for each $\sigma$,
 $$  \bar{\mathcal{P}}_\mathcal{T}(\sigma) \le  \mathcal{P}_\mathcal{T}(\sigma) \le  \bar{\mathcal{P}}_\mathcal{T}(\sigma)  + 2 \varepsilon. $$
By imposing $\sigma=\sigma^\star$ we get
$ \mathcal{P}^\star_\mathcal{T}=  \mathcal{P}_\mathcal{T}(\sigma^\star) \le  \bar{\mathcal{P}}_\mathcal{T}(\sigma^\star)  + 2\varepsilon \le  \bar{\mathcal{P}}_\mathcal{T}(\bar{\sigma}^\star)  + 2\varepsilon .$
By imposing $\sigma=\bar{\sigma}^\star$ we get
$ \bar{\mathcal{P}}_\mathcal{T}(\bar{\sigma}^\star) \le  \mathcal{P}_\mathcal{T}(\bar{\sigma}^\star).$
Combining the last two inequalities we get the desired bound. The last result can be proven as in Proposition \ref{tangent2}. Note that $\bar P$ is a vector of probabilities, hence we can set $\bf{M}=\mathbbm{1}$.
\end{proof}

\section{The gene expression network case study}
\label{gene}
To illustrate our method we consider again the gene expression model of Example \ref{ex:gene} and determine what combinations of the protein mean and variance are achievable starting from the zero state, under different assumptions on the  external signal. 

\subsection{ Single input}
 \label{gene_one}
Consider the gene expression model  with  one external signal and reactions following the mass action kinetics, as described in Example \ref{ex:gene1}. In this case, the moments equations are linear and the protein mean and variance can  be obtained by assuming as output matrix for the linear system~\eqref{eq:gene_1}
 \begin{align*}
L&:=\left[\begin{array}{ccccc}0 & 1 & 0  &0& 0\\ 0 &  0&0 & 0 & 1  \ \end{array}\right].
\end{align*}
Depending on the experimental setup, the external signal $\sigmau(t)$  may take values in the set $\Sigma^d:=\{0,1\}$, if the  input is of the ON-OFF type \cite{milias2011silico, parise2015guides, Olson2014a,menolascina2011analysis},  or  in the interval $\Sigma^c:=\left[0,1\right]$, if the  input is continuous \cite{Batt2012}. Corollary~\ref{cor:linear} guarantees the validity of the following results both for $\Sigma^d$ and $\Sigma^c$. 
The problem of computing an outer approximation of the reachable set  of this system~was  studied in \cite{parise2014reachable} using ad hoc methods. In Fig.~\ref{fig:silico} we compare the outer approximation obtained  therein (magenta dashed/dotted line) with the inner (solid red) and outer (dashed blue) approximations that we obtained using the methods for linear moment equations of Section \ref{sec:r_affine_af}. 
We used the parameters 
$
k_r=0.0236,   \gamma_r= 0.0503,  k_p= 0.18, \gamma_p=  0.0121
$
(all in units of  min$^{-1}$)  and set $T=360$ min. Figure \ref{fig:silico}  shows that the outer approximation computed using the hyperplane method is more accurate than the one previously obtained in the literature. Moreover, since inner and outer approximations practically coincide, this method allows one to effectively recover the  reachable set. 
\begin{figure}[h]
\begin{center}
\includegraphics[width=0.45\textwidth]{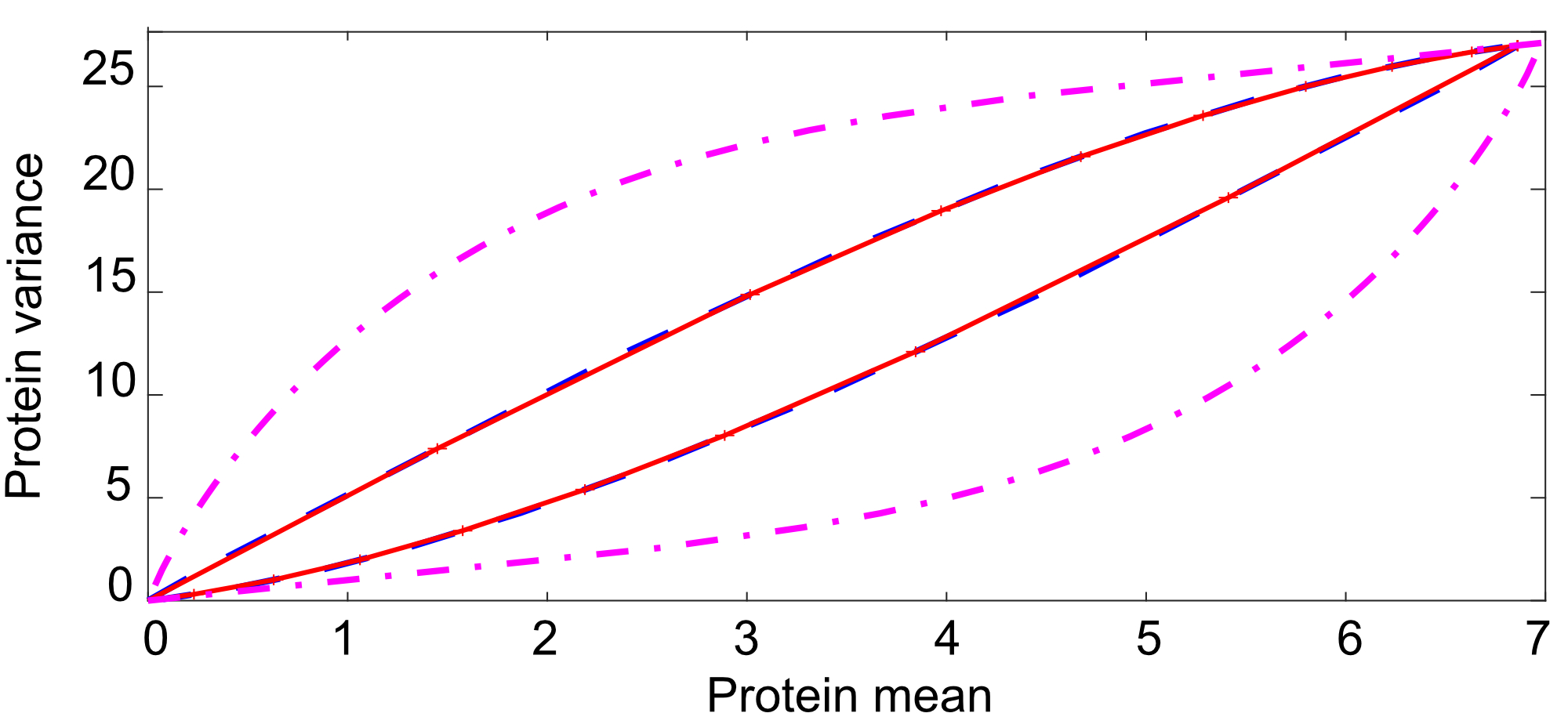}
\end{center}\vspace{-0.5cm} 
\caption{Comparison of the inner (red solid) and outer (blue dashed) approximations of the reachable set for the protein mean and variance, according to model \protect \eqref{eq:gene_1}, computed using the hyperplane method and the outer approximation computed according to \cite{parise2014reachable}  (magenta dashed/dotted). }
\label{fig:silico}
\end{figure}

\subsection{Single input and saturation}
\label{sat}

As second case study we consider again  Example \ref{ex:gene1}, but we now assume that not all the reactions follow the laws of mass action kinetics. Specifically, we are interested in investigating how the reachable set changes if we assume that the number of ribosomes in the cell is limited and consequently we impose a saturation to the translation propensity. Following \cite{brockmann2007posttranscriptional}, we assume that the translation rate follows the Michaelis-Menten kinetics so that 
$$\textstyle \alpha_3(k_p,z)= \tilde{k}_p \cdot \frac{a \cdot m}{b+a \cdot m} \ \mbox{ instead of }\  \alpha_3(k_p,z)= k_p \cdot m.$$
For the simulations we impose $\tilde{k}_p=0.7885, b=0.06, a=0.02$, so that the maximum reachable protein mean is the same as in the case without saturation analysed in the previous subsection. The corresponding propensity function is illustrated  in Fig.  \ref{fig:sat}a).
All the other propensities are assumed as in Section~\ref{gene_one}. Note that in this case the propensities are not affine. Consequently, we estimate the reachable set by  using the FSP approach derived in Theorem~\ref{krogh_inf}. Specifically we consider as set $J$ the indices corresponding to states with less than $6$ mRNA copies and $40$ protein copies. By assuming $T=360$  min and that $u$ can switch any $30$ minutes in the set $\Sigma^d=\{0,1\}$, we obtain an error $\varepsilon=2.84\cdot 10^{-4}$. Fig. \ref{fig:sat}b) shows the comparison of the reachable sets obtained for the cases with and without saturation. From this plot it emerges that, for the chosen values of parameters, saturation leads to a decrease of variability in the population.

\begin{figure}[h]
\quad \textbf{a)}
\vspace{-0.7cm}
\begin{center}
\hspace{0.35cm} \includegraphics[width=0.44\textwidth]{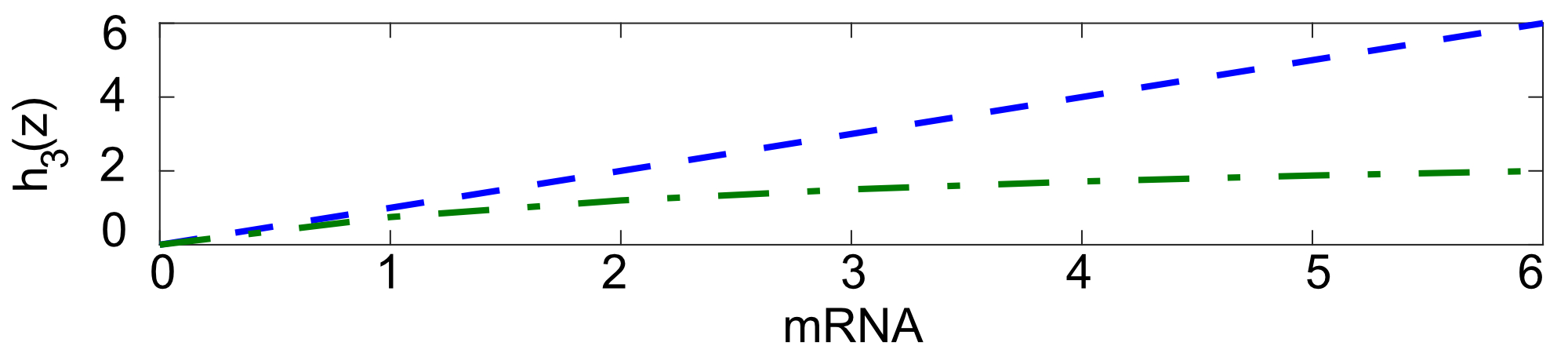}
\end{center}
\quad \textbf{b)}
\vspace{-0.9cm}
\begin{center}
\hspace{0.3cm} \includegraphics[width=0.45\textwidth]{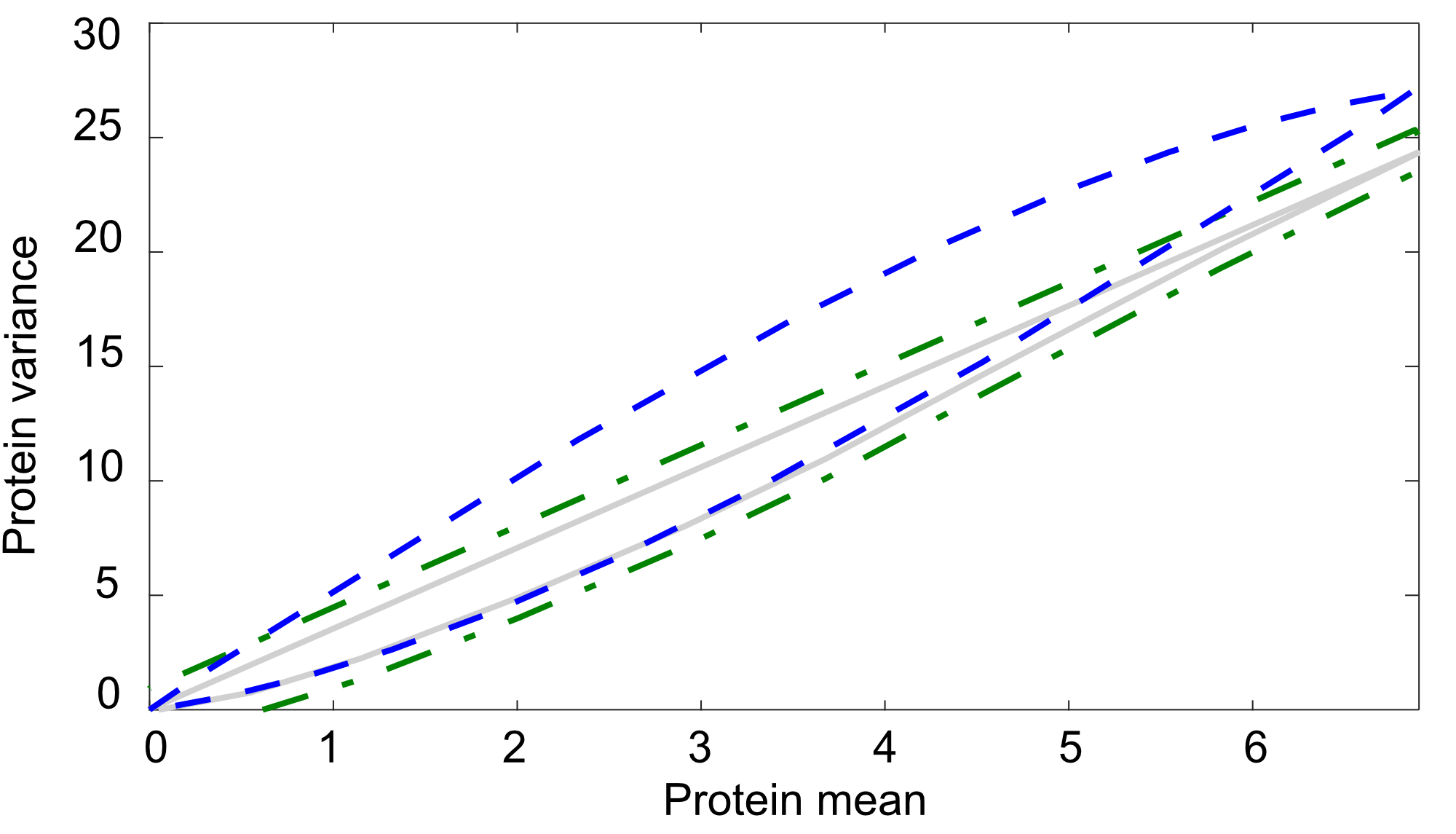}
\end{center}
\vspace{-0.5cm} 
\caption{\scriptsize Comparison of  the reachable set for the protein mean and variance, according to the model \protect in Example 2, when all the reactions follows the mass action kinetics (as in Fig \ref{fig:silico}) and when the translation is saturated. \textit{Figure a):} $h_3(z)$ when the translation reaction follows the mass action kinetics (dashed blue) or the Michaelis Menten kinetics (dashed dotted green). \textit{Figure b):} comparison of the outer approximations of the reachable sets in the two cases. The blue dashed line is as in Fig \ref{fig:silico}. The grey line is the outer approximation of the reachable set of the FSP system~\eqref{eq:lin_fin}, the green dashed dotted line is the outer approximation of the original system~\eqref{eq:lin_inf}  according to Theorem \ref{krogh_inf}. }
\label{fig:sat}
\end{figure}

\subsection{Fluorescent protein and the two  inputs case }

Consider again Example \ref{ex:gene}, but  now assume  that:
\begin{enumerate}
\item
  mRNA production and degradation can both be controlled, so that the vector of propensities is
$\alpha(z)=[k_r \cdot \sigmau_1(t),\ \gamma_r \cdot m\cdot \sigmau_2(t),\ k_p\cdot m,\ \gamma_p\cdot p]^\top$ and $\sigmau(t):=\begin{ps}\sigmau_1(t)\\\sigmau_2(t)\end{ps}$;
\item the protein $P$ can mature into a fluorescent protein $F$ according to the additional maturation and degradation reactions 
\begin{align*}
P \quad & \xrightarrow{\alpha_5(k_f, z)} \quad F , \quad\quad\quad\quad F \quad  \xrightarrow{\alpha_6(\gamma_p,z)} \quad  \emptyset,
\end{align*} 
where $\alpha_5(k_f, z):=k_f \cdot p,\quad \alpha_6(\gamma_p,z):=\gamma_p \cdot f$ and $k_f>0$ is the maturation rate. For simplicity, the degradation rate of  $F$ is assumed to be the same as that of $P$;
\item the fluorescence intensity $I(t)$ of each cell can be measured and is proportional to the amount of fluorescence protein, that is, $I(t)=rF(t)$ for a fixed scaling parameter $r>0$.
\end{enumerate}
Since all the propensities are affine, the system~describing the evolution of means and variances of the augmented network is 
\bev
\dot x_{\le 2}(t)=A^f({\sigmau(t)}) x_{\le 2}(t)+b^f({\sigmau(t)}),
\label{fluo}
\eev
where the state vector $ x_{\le 2}(t)$ and $A^f({\sigmau(t)}) ,b^f({\sigmau(t)})$ are 
 { \small \begin{align*}
 & x_{\le 2}\!\!=\left[\mathbb{E}[M] , \mathbb{E}[P] , \mathbb{E}[F]  , \Cov\![M,P], \Cov\![M,F] , \Var[P] , \Cov\![P,F] , \Var[F]\right]^\top\\[0.1cm]
&A^f\!\!=\!\!\left[ \arraycolsep=2.5pt\begin{matrix} d_1( \sigmau_2(t))  & 0 & 0 & 0 & 0 & 0 & 0 &0 \cr
 k_p&  d_2 &0 & 0& 0 &0 &0& 0 \cr
    0 &\gamma_p &  d_3&0& 0& 0& 0& 0 \cr
    k_p &0& 0 & d_4( \sigmau_2(t)) &0 &0& 0 &0 \cr
    0& 0& 0& \gamma_p&  d_5( \sigmau_2(t))  &0 &0 &0 \cr
    k_p&(\gamma_p+k_f)& 0 &2k_p &0 & d_6 &0& 0 \cr
    0 &-\gamma_p& 0& 0& k_p& \gamma_p& d_7&0 \cr
    0 &\gamma_p &k_f& 0 &0& 0& 2\gamma_p&  d_8
    \end{matrix}\right], \\
&b^f\!= \left[\arraycolsep=1.8pt\begin{matrix}
k_r \sigmau_1(t)&  \hspace{0.9cm} 0 &  \hspace{0.7cm} 0 & \hspace{0.7cm}0 & \hspace{1.15cm}0 & \hspace{0.7cm}0 & \hspace{0.3cm}0 & \hspace{0.2cm}0  \end{matrix}\hspace{0.2cm} \right]^\top\\[0.1cm] 
 &{\normalsize\mbox{with }} d_1( \sigmau_2(t)) =-\gr  \sigmau_2(t),\ d_2= -(\gamma_p+k_f),\ d_3=-k_f , \\& d_4( \sigmau_2(t)) =-(\gamma_r \sigmau_2(t)+\gamma_p+k_f), d_5( \sigmau_2(t)) =-(\gamma_r \sigmau_2(t)+k_f),\\&d_6=-2(\gamma_p+k_f),\ d_7= -(2k_f+\gamma_p),\ d_8=-2k_f.
\end{align*}}
System \eqref{fluo} depends on the parameter vector $\theta=[k_r, \gamma_r, k_p, \gamma_p, k_f,r]$ (for more details see \cite[Supplementary Information pg. 16]{Ruess2013}).  For the parameters we use the MAP estimates identified in \cite{parise2015reachable} (all in min$^{-1}$)
\begin{equation}
 \begin{array}{rl}
k_r&=0.0236 \quad   \gamma_r= 0.0503 \quad  k_p= 178.398  \cr
  k_f&= 0.0212 \quad \gamma_p=  0.0121\quad r^{-1}=646.86
  \end{array}
\label{param}
\end{equation}
and we  set
\begin{align}\label{out_fluo}
L^f&:=\left[\begin{array}{ccccccccc}0 & 0 & r   & 0 & 0 & 0 & 0&0\\ 0 & 0 &  0 & 0 & 0 & 0 & 0&r^2 \ \end{array}\right],
\end{align}
to compute the mean and variance reachable set for the fluorescence intensity. 

Our first aim is to compare the reachable set of such extended model with experimental data, when only one external signal (``1in") is available. In the case of one input, \eqref{fluo} is a linear system~and the methods of Section \ref{sec:r_affine_af} can be applied. Fig. \ref{fig:all}a)  shows the estimated reachable set compared with the real data collected in \cite{parise2015guides}.

\begin{figure}

\quad \textbf{a)}
\vspace{-0.6cm}
\begin{center}
\includegraphics[width=0.42\textwidth]{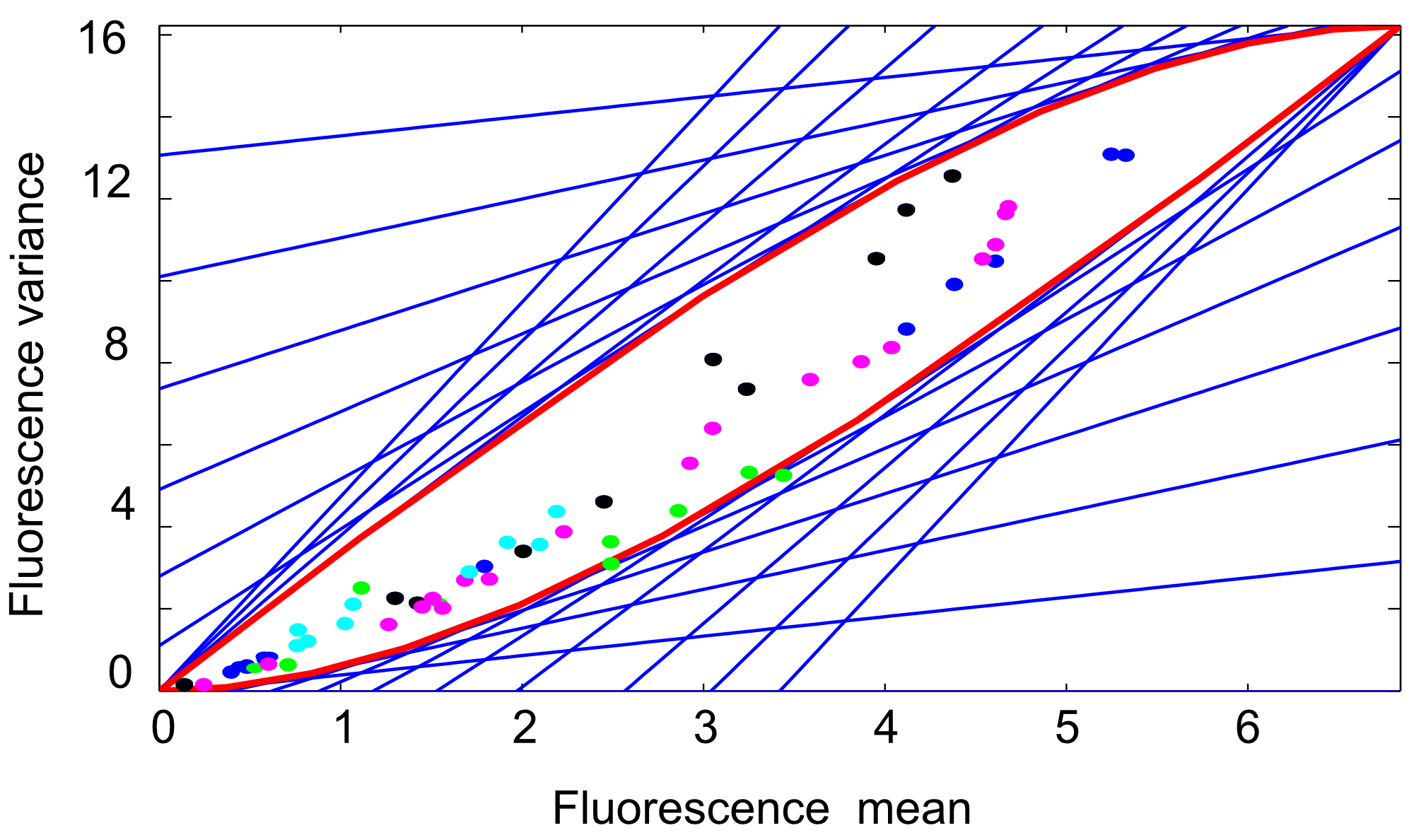}
\end{center}
\quad\textbf{b)}
\vspace{-0.7cm}
\begin{center}
\includegraphics[width=0.42\textwidth]{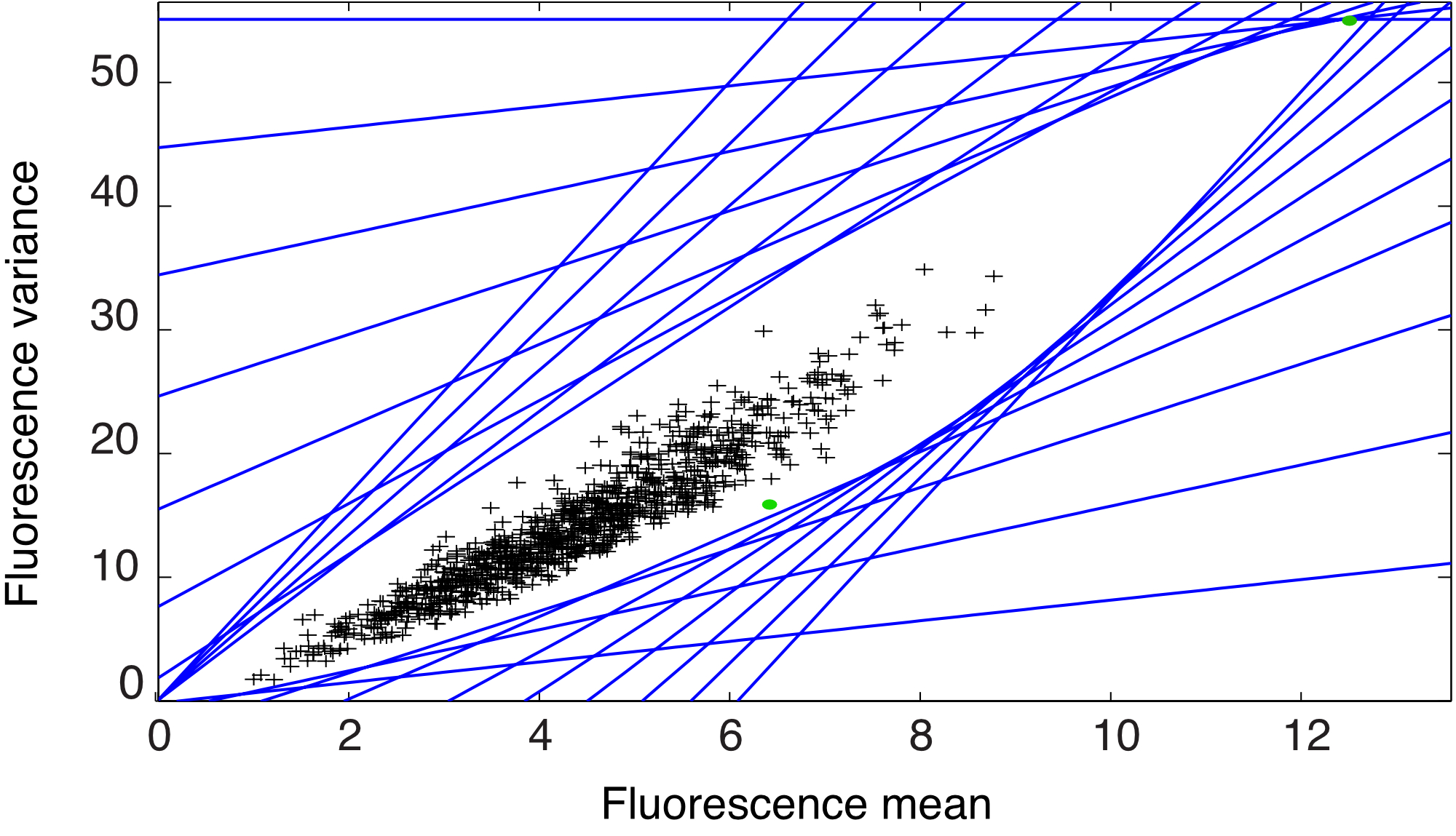}
\end{center}
\quad\textbf{c)}
\vspace{-0.8cm}
\begin{center}
\includegraphics[width=0.43\textwidth]{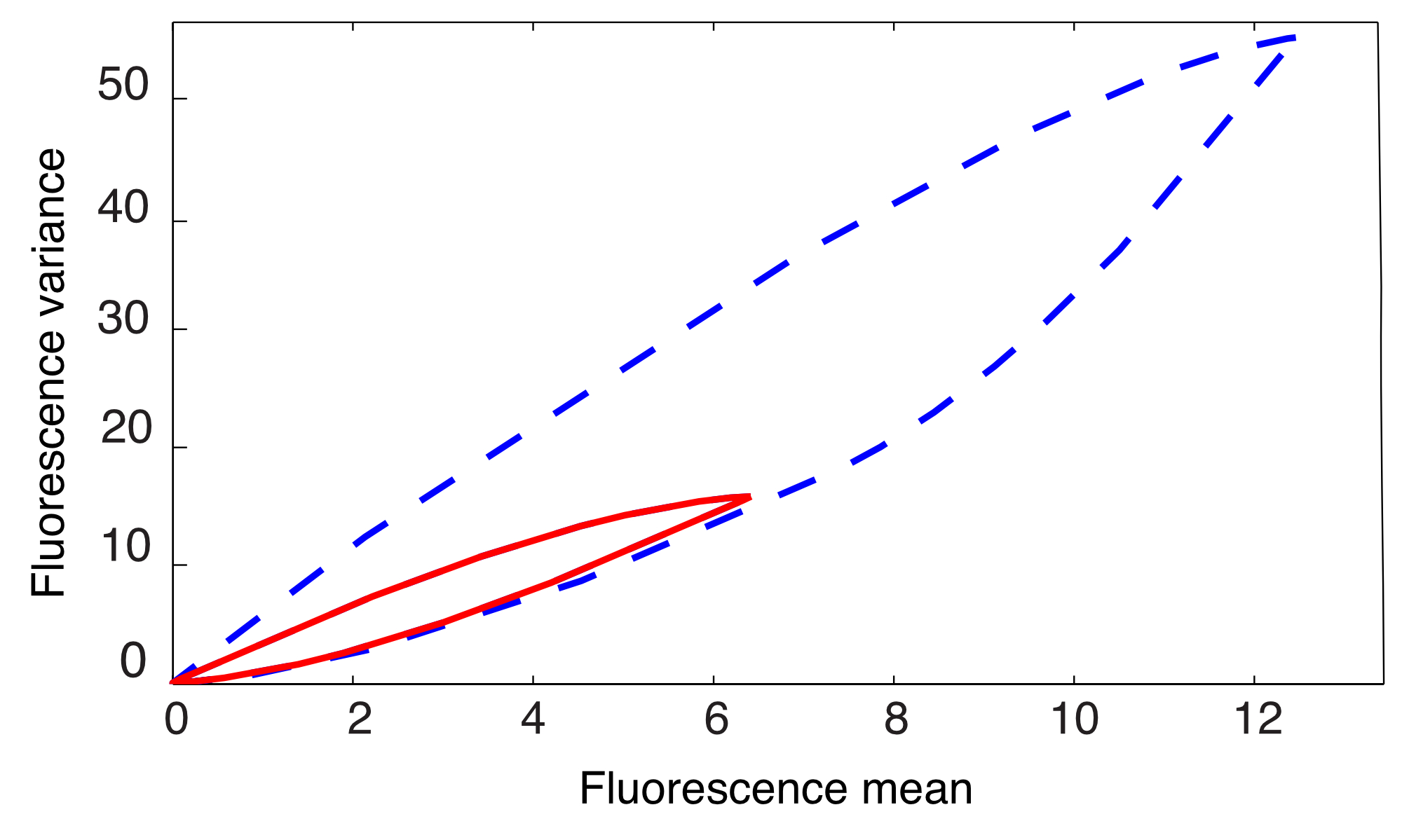}
\end{center}
\vspace{-0.5cm} 
\caption{ \scriptsize Output reachable set of system~\protect\eqref{fluo} with output as in \protect\eqref{out_fluo} and parameters as in \protect\eqref{param}. \textit{Figure a) [$1$ external signal]:} Comparison between the inner (red contour) and outer (blue lines) approximation of the output reachable set, when the set of possible modes is $\Sigma^{\textup{1in}}$, and the measured data. Different colors refers to data collected in different experiments. \textit{Figure b) [$2$ external signals]:}  Outer approximation of the output reachable set, when the set of possible modes is $\Sigma^{\textup{2in}}$. The two green dots represent the outputs when $\sigmau(t)=[1, 0.5]^\top\, \forall t$ and $\sigmau(t)=[1, 1]^\top\, \forall t$, respectively. The black crosses  represent the output for random signals in $\Sigma^{\textup{2in}}$.
\textit{Figure c) [Comparison]:} The red solid line  is the outer approximation obtained for $\sigmau(t) \in \Sigma^{\textup{1in}}$, as in Fig.~a), the blue dashed line the one for  $\sigmau(t)\in \Sigma^{\textup{2in}}$, as in Fig. b).}
\label{fig:all}
\end{figure}

Our second goal is to investigate how the reachable set changes when  both mRNA production and degradation are controlled (``2in"), as studied in \cite{briat2012computer}. Note that in this case, system~\eqref{fluo} is nonlinear.  We therefore set $T=300$  min and  assume that switchings can occur every $20$ min, so that Assumption~\ref{ass:input_switch} is satisfied with $K=15$ and use the hyperplane method as described in Section \ref{sec:r_affine_sw} with input  sets
\begin{eqnarray*}
\Sigma^{\textup{2in}}\!\!\!\!&:=& \!\!\!\! \left\{\left[\begin{array}{c}0 \\1 \end{array}\right], \left[\begin{array}{c}0 \\0.5 \end{array}\right], \left[\begin{array}{c}1 \\1 \end{array}\right],\left[\begin{array}{c}1 \\0.5 \end{array}\right]\right\}\!\!, \mbox{ so that } I=4,\\
\Sigma^{\textup{1in}}\!\!\!\!&:=&\!\!\!\!\left\{\left[\begin{array}{c}0 \\1 \end{array}\right], \left[\begin{array}{c}1 \\1 \end{array}\right] \right\}, \mbox{ so that } I=2,
\end{eqnarray*}
respectively. Note that we set the minimum input for the mRNA degradation to $0.5>0$ to avoid unboundedness. With these input choices it is intuitive that the largest possible state is reached when the mRNA production is at its maximum and the mRNA degradation is at its minimum. Therefore, in the {MILP}s we can use the bounds ${\bf M}=x\left(T;0, \sigmau(t)=\begin{ps}1\\0.5\end{ps}\, \forall t\right)$ for the case of two inputs and ${\bf M}=x(T;0,  \sigmau(t)=\begin{ps}1\\1\end{ps}\, \forall t), $ for the case of one input. 
Fig.~\ref{fig:all}b) shows the output reachable set for the case of two inputs.
The simulation time for computing the outer approximation with the hyperplane method was $5.6$ hrs. Computing the exact reachable set by simulating all the possible switching signals, assuming that one simulation takes $10^{-4}$ sec and neglecting the time needed to enumerate all   possible signals, would take $29.8$ hrs.  The black crosses in Fig.~\ref{fig:all}b) are obtained by simulating the output of the system~for $5000$ randomly constructed input signals. This simulation illustrates that random approaches might lead to significantly under estimate the  reachable set.
Fig.~\ref{fig:all}c) shows a comparison of the reachable sets obtained in Fig.~\ref{fig:all}a) and b) when the input set is $\Sigma^{\textup{1in}}$ and $\Sigma^{\textup{2in}}$, respectively.

\section{Conclusion}

In the paper we have: i) proposed a method to approximate the projected reachable set of switched affine 
systems with fixed switching times,   ii) extended the FSP approach  to controllable  networks, iii)
illustrated how these new theoretical tools can be used to analyse generic   networks both from a population and single cell perspective and iv) provided an extensive gene expression case study using both in silico and in vivo data.
Even though our analysis is motivated by biochemical reaction networks, 
 our results  can actually be applied to study the  moments of any Markov chain with  transitions rates that switch among $I$ possible configurations at $K$ fixed instants of times. Our results hold both in case of finite and infinite state space.
Moreover, while
 we have assumed here that cells are identical, we showed in \cite{parise2016reachability} that also in the case of heterogeneous population one can derive equations describing the moments evolution. The reachable set of such populations can  be obtained, as described in this paper, by applying Corollary~\ref{cor:linear} or \ref{thm:hm1_sw} to such system.


\enlargethispage{-8cm}
\vspace{-5cm}
\begin{IEEEbiography}[{\includegraphics[trim=40mm 0mm 50mm 10mm,width=1.25in,height=1.25in,clip,keepaspectratio]{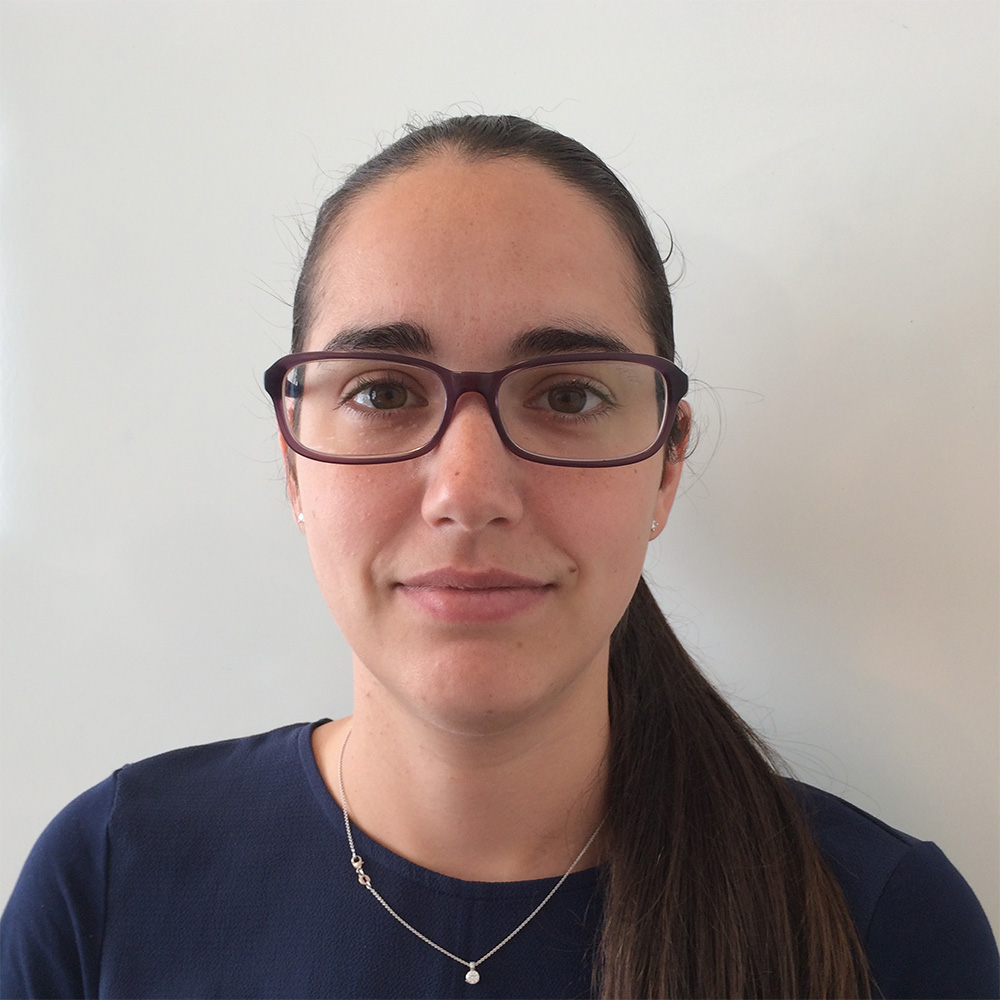}}]{Francesca Parise}
 Francesca Parise was born in Verona, Italy, in 1988. She received the B.Sc. and M.Sc. degrees (cum Laude) in Information and Automation Engineering from the University of Padova, Italy, in 2010 and 2012, respectively. She conducted her master thesis research at Imperial College London, UK, in 2012.  She graduated from the Galilean School of Excellence, University of Padova, Italy, in 2013. She defended her PhD at the Automatic Control Laboratory, ETH Zurich, Switzerland in 2016 and she is currently a Postdoctoral researcher at the Laboratory for Information and Decision Systems, M.I.T., USA.

Her research focuses on identification, analysis and control of complex systems, with application to distributed multi-agent networks and systems biology.
\end{IEEEbiography}

\vspace{-5cm}
\begin{IEEEbiography}[{\includegraphics[trim=18mm 0mm 20mm 0cm,width=1.25in,height=1.25in,clip,keepaspectratio]{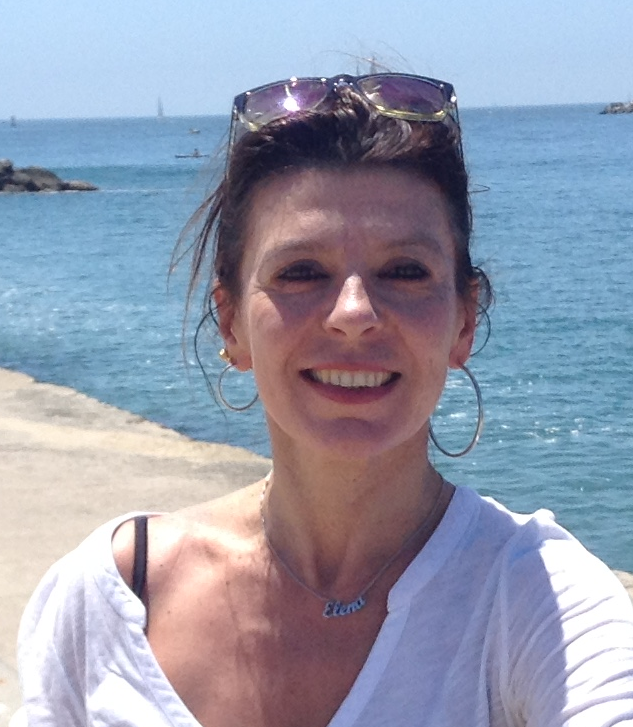}}]{Maria Elena Valcher}
 Maria Elena Valcher  received the Laurea degree and the PhD
from the University of Padova, Italy.
Since January 2005  she is full professor of Automatica at the University of Padova.

She is author/co-author of 76 papers appeared in international journals, 92 conference papers, 2 text-books and several book chapters. Her research interests include multidimensional systems theory, polynomial matrix theory, behavior theory, convolutional coding, fault detection, delay-differential systems, positive systems, positive switched systems and Boolean control networks.

She has been involved in the Organizing Committees and in the Program Committees of several conferences. In particular, she was the Program Chair of the CDC 2012. 
She was in the Editorial Board of the IEEE Transactions on Automatic Control (1999-2002), Systems and Control Letters (2004-2010) and she is currently in the Editorial Boards of Automatica (2006-today), Multidimensional Systems and Signal Processing (2004-today), SIAM J. on Control and Optimization (2012-today),  European Journal of Control (2103-today) and IEEE Access (2014-today).

She was Appointed Member of the CSS BoG (2003); Elected Member of the CSS BoG (2004-2006; 2010-2012); Vice President Member Activities of the CSS (2006-2007); Vice President Conference Activities of the CSS  (2008-2010); CSS President (2015).
She is presently a Distinguished Lecturer of the IEEE CSS and  IEEE CSS Past President. She received the 2011 IEEE CSS Distinguished Member Award and she is an IEEE Fellow since 2012.
\end{IEEEbiography}
\vspace{-2cm}
\begin{IEEEbiography}[{
\includegraphics[width=1.25in,height=1.25in,clip,keepaspectratio]{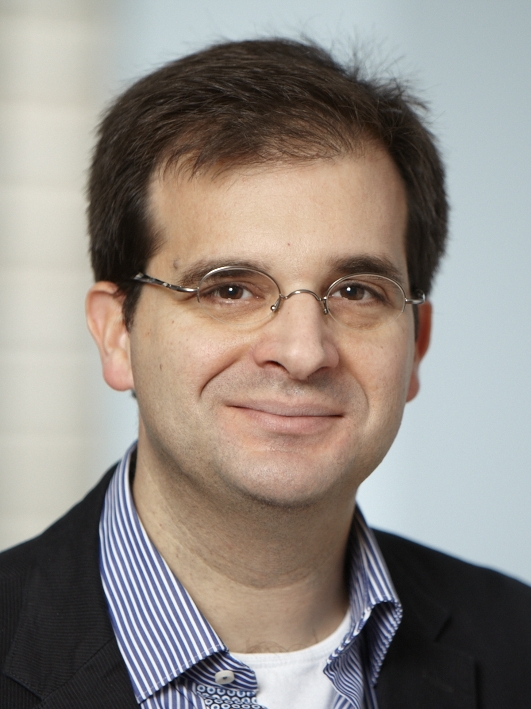}}]{John Lygeros}
  John Lygeros completed a B.Eng. degree in electrical engineering in  1990 and an M.Sc. degree in Systems Control in 1991, both at Imperial College of Science Technology and Medicine, London, U.K.. In 1996 he obtained a Ph.D. degree from the Electrical Engineering and Computer Sciences Department, University of California, Berkeley. During the period 1996-2000 he held a series of research appointments at the National Automated Highway Systems Consortium, Berkeley, the Laboratory for Computer Science, M.I.T., and the Electrical Engineering and Computer Sciences Department at U.C. Berkeley. Between 2000 and 2003 he was a University Lecturer at the Department of Engineering, University of Cambridge, U.K., and a Fellow of Churchill College. Between 2003 and 2006 he was an Assistant Professor at the Department of Electrical and Computer Engineering, University of Patras, Greece. In July 2006 he joined the Automatic Control Laboratory at ETH Zurich, first as an Associate Professor, and since January 2010 as a Full Professor. Since 2009 he is serving as the Head of the Automatic Control Laboratory and since 2015 as the Head of the Department of Information Technology and Electrical Engineering. His research interests include modelling, analysis, and control of hierarchical, hybrid, and stochastic systems, with applications to biochemical networks, automated highway systems, air traffic management, power grids and camera networks. He teaches classes in the area of systems and control both at the undergraduate and at the graduate level at ETH Zurich, notably the 4th semester class Signals and Systems II, which he delivers in a flipped classroom format. John Lygeros is a Fellow of the IEEE, and a member of the IET and the Technical Chamber of Greece. He served as an Associate Editor of the IEEE Transactions on Automatic Control and in the IEEE Control Systems Society Board of Governors; he is currently serving as the Treasurer of the International Federation of Automatic Control.
\end{IEEEbiography}

\end{document}